\documentclass[%
 reprint,
superscriptaddress,
 amsmath,amssymb,
 aps,
pra,
]{revtex4-1}

\usepackage{amsthm,amssymb}
\usepackage{amsmath}
\usepackage{graphicx}
\usepackage{dcolumn}
\usepackage{bm}
\usepackage{multirow} 
\usepackage{mathdots}
\usepackage{enumerate}
\usepackage{algorithm}
\usepackage{algorithmic}
\usepackage{color}
\usepackage{CJK}
\usepackage{float}
\usepackage{makecell}
\usepackage{booktabs}
\usepackage{adjustbox}

\theoremstyle{remark}
\newtheorem{definition}{\indent Definition}
\newtheorem{lemma}{\indent \emph{\textbf{Lemma}}}
\newtheorem{theorem}{\indent \emph{\textbf{Theorem}}}
\newtheorem{corollary}{\indent \textbf{\emph{Corollary}}}

\def\proof{\indent{\em Proof}.}

\begin{document}

\makeatletter
\newcommand{\ud}{\mathrm{d}}
\newcommand{\rmnum}[1]{\romannumeral #1}
\newcommand{\Rmnum}[1]{\expandafter\@slowromancap\romannumeral #1@}
\newcommand{\udots}{\mathinner{\mskip1mu\raise1pt\vbox{\kern7pt\hbox{.}}
        \mskip2mu\raise4pt\hbox{.}\mskip2mu\raise7pt\hbox{.}\mskip1mu}}
\makeatother

\preprint{APS/123-QED}

\title{Block-encoding based quantum algorithm for linear systems with displacement structures}
\author{Lin-Chun Wan}
\affiliation{State Key Laboratory of Networking and Switching Technology, Beijing University of Posts and Telecommunications, Beijing, 100876, China}
\affiliation{State Key Laboratory of Cryptology, P.O. Box 5159, Beijing, 100878, China}
\author{Chao-Hua Yu}
\affiliation{The School of Information Management, Jiangxi University of Finance and Economics, Nanchang, 330032, China}
\author{Shi-Jie Pan}
\affiliation{State Key Laboratory of Networking and Switching Technology, Beijing University of Posts and Telecommunications, Beijing, 100876, China}
\author{Su-Juan Qin}
\email{qsujuan@bupt.edu.cn}
\affiliation{State Key Laboratory of Networking and Switching Technology, Beijing University of Posts and Telecommunications, Beijing, 100876, China}
\author{Fei Gao}
\email{gaof@bupt.edu.cn}
\affiliation{State Key Laboratory of Networking and Switching Technology, Beijing University of Posts and Telecommunications, Beijing, 100876, China}
\author{Qiao-Yan Wen}
\email{wqy@bupt.edu.cn}
\affiliation{State Key Laboratory of Networking and Switching Technology, Beijing University of Posts and Telecommunications, Beijing, 100876, China}

\date{\today}

\begin{abstract}
Matrices with the displacement structures of circulant, Toeplitz, and Hankel types as well as matrices with structures generalizing these types are omnipresent in computations of sciences and engineering.  In this paper, we present efficient and memory-reduced quantum algorithms for solving linear systems with such structures by devising a new approach to implement the block-encodings of these structured matrices. More specifically, by decomposing $n\times n$ dense matrices into linear combinations of displacement matrices, we first deduce the parameterized representations of the matrices with displacement structures so that they can be treated similarly. With such representations, we then construct $\epsilon$-approximate block-encodings of these structured matrices in two different data access models, i.e., the black-box model and the QRAM data structure model. It is shown the quantum linear system solvers based on the proposed block-encodings provide a quadratic speedup with respect to the dimension over classical algorithms in the black-box model and an exponential speedup in the QRAM data structure model. In particular, these linear system solvers subsume known results with significant improvements and also motivate new instances where there was no specialized quantum algorithm before.  As an application, one of the quantum linear system solvers is applied to the linear prediction of time series, which justifies the claimed quantum speedup is achievable for problems of practical interest.
\end{abstract}
%
%
%
%
%
\pacs{Valid PACS appear here}
\maketitle
\section{Introduction}
Quantum technologies have shown their significant influence in communication and computing. On the one hand, many quantum cryptographic protocols have been proposed for protecting the security and privacy \cite{bennett2014quantum,wei2020error,xu2021novel,gao2019quantum,fitzsimons2017private}. On the other hand, quantum computing which makes use of quantum mechanical principles, such as superposition and entanglement, shows tremendous potential that outperforms the conventional computing in time complexity in solving many problems, including Boolean function computing \cite{he2018exact,chen2020characterization}, matrix computing \cite{eq1,zhao2021compiling,gilyen2019quantum} and machine learning \cite{yu2016quantum,yu2019quantumprincipal,yu2019improved,pan2020improved,alg}.

Matrices encountered in practical computations often have some special structures, and many problems can be transformed into solving linear systems with structures. Among various matrix structures, the ones of circulant, Toeplitz, Hankel types and their generalization called circulant/Toeplitz/Hankel-like are best-known and well-studied \cite{kailath1999fast,pan2001structured,peller2012hankel}. As of now, some efficient quantum algorithms have been proposed for solving linear systems with these displacement structures. The quantum algorithm for the Poisson equation is the earliest work (as far as we know) on solving Toeplitz linear systems involving a specific kind of banded Toeplitz matrices \cite{pos}. In 2016, a quantum algorithm for solving sparse circulant systems was proposed by Mahasinghe and Wang \cite{Tb}. Whereafter, the quantum algorithm for solving circulant systems with the bounded spectral norm was presented in Ref. \cite{zhou2017efficient}, where no assumption on the sparseness is demanded. By constructing associated circulant matrices, Wan et al \cite{wan2018asymptotic} proposed a quantum algorithm for solving Toeplitz systems in Wiener class (defined in Sec.\ref{seciiic}), which is an asymptotic quantum algorithm of which the error is related to the dimension of the Toeplitz matrices.

The quantum algorithms introduced above can achieve excellent performance under certain circumstances and are powerful for solving large families of problems. However, these algorithms followed different ideas and employed different techniques that the quantum algorithm for linear systems with a specific displacement structure is not inspiring for other types. Moreover, it is intractable to generalize these algorithms to solve linear systems with the same type of generalized structures. Then, an interesting question is can we design quantum algorithms for linear systems with various types of displacement structures in a unified way. A unified treatment of such structured matrices can often provide conceptual and computational benefits and may also give insight into finding some new applicable instances. Combined with the method of block-encoding, we answer the question in the affirmative.

A block-encoding of a matrix $\boldsymbol M$ is a unitary $\boldsymbol U$ that encodes $\boldsymbol M/\alpha$ as its top left block, where $\alpha\geq \|\boldsymbol M\|$ is a scaling factor. Given a way to implement block-encodings of some matrices, many operations on the matrices can be done, including linear system solving \cite{gilyen2019quantum,low2019hamiltonian,chakraborty2019power}. Nevertheless, it is worth noting the complexity of the quantum algorithm for linear systems based on block-encodings has a linear dependence on scaling factors, and implementing block-encodings with preferred scale factors often requires ingenious design. Although there are some methods to implement the block-encodings for several specific matrices, such as sparse matrices \cite{gilyen2019quantum,low2019hamiltonian}, density operators \cite{low2019hamiltonian,van2019improvements}, POVM operators \cite{van2019improvements}, Gram matrices \cite{gilyen2019quantum}, matrices stored in a quantum-accessible data structure \cite{chakraborty2019power,kerenidis2020quantum}, directly applying the methods mentioned can not give rise to appealing quantum linear system solvers for the matrices with displacement structures. Exploiting the structures of such matrices to implement their block-encodings with favorable scaling factors and less cost of time, space or memory deserves specialized study.

In this paper, we devise an approach that decomposes $n \times n$ dense matrices into linear combinations of unitaries (LCU), and implement block-encodings of matrices with displacement structures following the idea of LCU lemma \cite{Childs2017Quantum}. The proposed block-encodings can give rise to efficient quantum algorithms for a set of linear systems with displacement structures, including the linear systems, such as Toeplitz systems in Wiener class and circulant systems with the bounded spectral norm, whose quantum algorithms can be improved by our method, and the linear systems without specialized quantum algorithm before, such as some Toeplitz/Hankel-like linear systems. More specifically, the main contributions of this paper are as follows:

\begin{itemize}
\item[1)]
We first deduce parameterized representations of $n\times n$ dense matrices by decomposing them into linear combinations of unitaries, which provide a way that the structured matrices of interest can be represented and treated similarly. The proposed LCU decompositions possess several desirable features for implementing block-encodings. (i) The elementary components unitaries are displacement matrices that can be easily implemented; (ii)  The decomposition coefficients are the elements of the displacement of the decomposed matrices that can be easily calculated; (iii) For the structured matrices of interest, the number of decomposed items is roughly $O(n)$. Especially, this decomposition method provides a representation with $2n-1$ parameters for Toeplitz or Hankel matrices, which is optimal in terms of the number of parameters.
\end{itemize}

\begin{itemize}
\item[2)]
Based on the proposed LCU decompositions, we then construct efficient quantum circuits in two different data access models commonly used in various quantum algorithms, i.e., the black-box model and the QRAM data structured model, to implement the $\epsilon$-approximate block-encodings of matrices with displacement structures. If a matrix is given in the QRAM data structure model, it will often lead to a low-complexity construction. Otherwise, the construction scheme in the black-box model may be adopted since it requires less on matrix storage. In both models, we implement block-encodings of which the scaling factors are proportional to the $l_{1}$-norm $\chi$ of the displacement of the structured matrices. For the structured matrices with small $\chi$, the linear system solvers based on the proposed block-encodings provide a quadratic speedup with respect to the dimension over classical algorithms in the black-box model and an exponential speedup in the QRAM data structure model.
\end{itemize}




\begin{itemize}
\item[3)]
We show that many matrices with displacement structures that are frequently encountered in various practical problems can harness the potential advantages of the proposed constructions. With the quantum linear system solver developed in the block-encoding framework, we obtain the following algorithms. (i) A quantum algorithm for Toeplitz linear systems in the Wiener class, which is an exact algorithm that the error is independent of the dimension of the Toeplitz matrices. It positively answers the open question raised in \cite{wan2018asymptotic} and can provide computational benefits when rigorous precision is required. (ii) A quantum algorithm for circulant linear systems with the bounded spectral norm, providing a quadratic improvement in the dependence on the condition number and an exponential improvement in the dependence on the precision over the quantum algorithm proposed in \cite{zhou2017efficient}. (iii) An efficient quantum algorithm for linear systems with Toeplitz/Hankel-like structures. In particular, we also obtain a quantum algorithm for banded Toeplitz/Hankel linear systems without the use of any black box or QRAM, which may be more convenient when constructing practical quantum circuits. At last, we show that the proposed quantum algorithm for Toeplitz linear systems can be used in linear prediction of time series, which provides a concrete example to illustrate that the quantum speedup is practically achievable.
\end{itemize}

\section{Preliminaries}

\subsection{The Matrices With Displacement Structure}
The matrices with displacement structures arise pervasively in many contexts.  Below, we display three popular classes of such structured matrices and their generalizations, which are also our focus in this paper.

A Toeplitz matrix $\boldsymbol T_n$ is a matrix of size $n\times n$ whose elements along each diagonal are constants. More clearly,
\begin{equation}\label{eq:defineT}
\boldsymbol T_n =
\left( \begin{array}{ccccc}
t_{0} & t_{-1} & t_{-2} & \cdots & t_{-(n-1)}\\
t_{1} & t_{0} & t_{-1} & \ddots & \vdots \\
t_{2} & t_{1} & t_{0} & \ddots & t_{-2}  \\
\vdots & \ddots & \ddots & \ddots & t_{-1}\\
\ t_{(n-1)} & \cdots & t_{2} & t_{1} & t_{0}\\
\end{array} \right),
\end{equation}
where $t_{i,k} = t_{i-k}$, $\boldsymbol T_n$ is determined by the sequence $\{t_j\}_{j=-(n-1)}^{n-1}.$

There is a common special case of Toeplitz matrix called circulant matrix whose every row is a right cyclic shift of the row above it:
\begin{equation}
\boldsymbol C_n =
\left( \begin{array}{ccccc}
c_{0}       & c_{n-1}     & c_{n-2}      & \cdots      & c_{1}\\
c_{1}   & c_{0}     & c_{n-1}      & \cdots    & c_{2}\\
c_{2}   &   c_{1}     & \ddots     & \ddots      &  \vdots  \\
\vdots      &           & \ddots     &             & c_{n-1}\\
\ c_{n-1}     & \cdots    &            & c_{1}   & c_{0}\\
\end{array} \right).
\end{equation}
Since the circulant matrix has some fantastic properties, it has been studied in specialty no matter in the classical or quantum setting.

Another representative class of matrices with displacement structures is the Hankel matrix. A matrix $\boldsymbol H_n$ is called Hankel matrix if it has the form
\begin{equation}
\boldsymbol H_n =
\left( \begin{array}{ccccc}
h_{0}       & h_{1}     & h_{2}      & \cdots      & h_{n-1}\\
h_{1}       & h_{2}     &            & \udots      & h_{n}\\
h_{2}       &           & \udots     & \udots      &  \vdots  \\
\vdots      & h_{n-1}   & \udots     &            & h_{2n-3}\\
\ h_{n-1}   & h_{n}     & \cdots     & h_{2n-3}    & h_{2n-2}\\
\end{array} \right).
\end{equation}
The entry $h_{i,k}(i,k=0,1,\cdots,n-1)$ of $\boldsymbol H_n$ is equal to $h_{i+k}$ for given sequence $\{h_j\}_{j=0}^{2n-2}$.  In other words, the skew-diagonals of a Hankel matrix are constants.

There are some natural generalizations of these structured matrices called Toeplitz/Hankel-like matrices (circulant-like matrices are regarded as a special case of Toeplitz-like matrices). A matrix is said to be Toeplitz/Hankel-like if there are a few elements on some diagonals/skew-diagonals of the matrix that are not equal to the others.

The computations with these structured matrices, both Toeplitz/Hankel matrices and Toeplitz/Hankel-like matrices, are widely applied in the various areas of sciences and engineering. For example, in time series analysis, the covariance matrices of weakly stationary processes are Toeplitz matrices, see \cite{Toe}. The visual tracking framework of \cite{henriques2014high} requires a base sample image to generate multiple virtual samples which correspond to circulant matrices. The solvability of certain classical interpolation problems is connected with Hankel matrices, see \cite{peller2012hankel}. The numerical solutions of some partial differential equations with mixed boundary conditions can be obtained by solving the Toeplitz-like linear systems generated by the discretization of the finite difference method, see \cite{smith1985numerical}. Other applications involve polynomial computations \cite{pan2001structured}, image restoration \cite{CGT}, machine learning \cite{ye2018deep}, compressed sensing\cite{haupt2010toeplitz}, and so on.

Compared to general matrices, the structures in these matrices can be exploited to perform algebraic operations, such as matrix-vector multiplication, inversion, and matrix exponential, with much less running time and memory space.  There are a number of classical methods that have been presented to solve the linear systems with such structures \cite{CGT}. However, the time complexity of these methods is $\Omega(n)$, it is still a hard task to tackle the problems with very large $n$ on a classical computer.
\subsection{The Framework of Block-encodings}
In this section, we review the framework of block-encodings introduced in  \cite{low2019hamiltonian,chakraborty2019power}.

\begin{definition}[Block-encoding]\label{def:Block-encoding}
Suppose that $\boldsymbol M$ is an $s$-qubit operator, $\alpha, \epsilon\in \mathbb{R}^+$ and $a\in \mathbb{N}$. Then we say that the $(s + a)$-qubit unitary $\boldsymbol U$ is an $(\alpha;a;\epsilon)$-block-encoding of $\boldsymbol M$, if
\begin{equation}
 \|\boldsymbol M-\alpha(\langle0|^{\otimes a}\otimes \boldsymbol I)\boldsymbol U(|0\rangle^{\otimes a}\otimes \boldsymbol I)\|\leq\epsilon.
\end{equation}
\end{definition}

Given a block-encoding $\boldsymbol U$ of a matrix $\boldsymbol M$, one can produce the state $\boldsymbol M|\psi\rangle/\|\boldsymbol M|\psi\rangle\|$ by applying $\boldsymbol U$ to a initial state $|0\rangle|\psi\rangle$. Low and Chuang \cite{low2019hamiltonian} presented Hamiltonian simulation algorithm under the framework of block-encodings by combining techniques qubitization and quantum signal processing, which can simulate sparse Hamiltonians with  optimal complexity. Taking this Hamiltonian simulation algorithm as a subroutine, Chakraborty et al. \cite{chakraborty2019power} developed several useful tools within the block-encoding framework such as singular value estimation and quantum linear system solver. In fact, they also point out that one can implement any smooth function of a Hamiltonian when given a block-encoding of this Hamiltonian by using the techniques developed in \cite{van2017quantum}. Furthermore, the method of block-encoding has been applied to the study of machine learning, and many quantum algorithms have been presented such as quantum clustering algorithm \cite{kerenidis2019q}, quantum classification algorithm \cite{shao2019quantum} and quantum algorithms for semidefinite programming problems \cite{van2019improvements,kerenidis2020quantumACM}.

Although the block-encoding can be applied to various algorithms for various computational problems, we will narrow our goal to the detail study of solving linear systems, and then analyze the improvements brought by the method proposed in this paper. Here we describe an informal version of the complexity result of the quantum algorithm for linear systems in the framework of block-encoding. Given a $(\alpha; a; \epsilon)$-block-encoding $\boldsymbol U$ of $\boldsymbol M$, there is a quantum algorithm that produces a state that is  $\epsilon$-close to $\boldsymbol M^{-1}|b\rangle/\|\boldsymbol M^{-1}|b\rangle\|$ in time $O(\kappa_{\boldsymbol M }(\alpha(a+T_U)\textrm{log}\frac{1}{\epsilon}+T_b))$,  where  $\kappa_{\boldsymbol M}$ is the condition number of $\boldsymbol M$, $T_U$ is the running time to implement $\boldsymbol U$, $T_b$  is the running time to prepare the state $|b\rangle$.  This result suggests that one need to efficiently construct block-encodings with small scaling factors.
\subsection{The Data Access Model}\label{seciic}


We now specify the data access model involved in the proposed quantum algorithms for solving linear system $\boldsymbol M\boldsymbol x=\boldsymbol b$.



For the right-hand-side vector $\boldsymbol b$, a unitary that produces the quantum state $|b\rangle=\sum_ib_i|i\rangle/\|\sum_ib_i|i\rangle\|$ is required in the quantum linear system solving algorithm. It is shown that some specialized algorithms can be used to generate $|b\rangle$ efficiently under certain conditions \cite{prepare,prepare1}, or our quantum algorithm may be used as a subroutine while $|b\rangle$ can be prepared by another part of a larger quantum algorithm.  Alternatively, if an efficiently-implementable state preparation procedure can not be provided, we assume $\boldsymbol b$ can be accessed in the same way as the coefficient matrix, which will be introduced below.




For the coefficient matrix $\boldsymbol M$,  we are given two different data access models which are most commonly used in various quantum algorithms. In the first data access model, the elements of matrix $\boldsymbol M\in\mathbb{C}^{n\times n}$ are accessed by a black box $\mathcal{O}_{\boldsymbol M}$ acting as
\begin{equation}
  \mathcal{O}_{\boldsymbol M}|i\rangle|k\rangle|0\rangle=|i\rangle|k\rangle|m_{i,k}\rangle\qquad i,k=0,1,\cdots,n-1.
\end{equation}
This model is often referred to as the black-box model \cite{berry2012black,sanders2019black}. The access operation can be made efficiently when $m_{i,k}$ are efficiently computable or a QRAM is provided.

Kerenidis and Prakash \cite{kerenidis2016quantum} introduced a different data access model called the QRAM data structure model. This data access model stores data in QRAM with a binary tree structure and allows access in superposition. When the data is given as an $n\times n$ matrix,  the matrix is stored in the binary trees by rows, and an additional binary tree is required to store the norms of the rows. Obviously, the memory requirement and the complexity of constructing this data structure must be $O(n^2)$.
Although many quantum algorithms in this model do not take the complexity of constructing data structure, reducing memory requirement and thereby reducing the complexity of constructing data structure has many practical implications.
\section{Methods and results}

\subsection{LCU Decomposition of Matrices}

In this section, we deduce parameterized representations of $n\times n$ matrices by decomposing them into linear combinations of unitaries, where the unitaries used as elementary components are easy to implement and the decomposition coefficients are easy to calculate. Without loss of generality, we assume that $n$ is always a power of two.

For a better understanding, we first introduce some necessary background information about matrix displacement.

\begin{definition} \label{definition:1}
For a given pair of operator matrices $(\boldsymbol{A},\boldsymbol{B})$, and a matrix $\boldsymbol M \in\mathbb{C}^{n\times n}$, the linear displacement operators $\mathcal{L}(\boldsymbol M):\mathbb{C}^{n\times n}\longmapsto \mathbb{C}^{n\times n}$ of Stein type is defined by:
\begin{equation}
\mathcal{L}(\boldsymbol M)=\Delta_{\boldsymbol A,\boldsymbol B}[\boldsymbol M]=\boldsymbol M-\boldsymbol {AMB},
\end{equation}
 and that of Sylvester type is defined by:
\begin{equation}
\mathcal{L}(\boldsymbol M)=\nabla_{\boldsymbol A,\boldsymbol B}[\boldsymbol M]=\boldsymbol {AM}-\boldsymbol{MB}.
\end{equation}
\end{definition}

The image $\mathcal{L}(\boldsymbol M)$ of the operator $\mathcal{L}$ is called the displacement of the matrix $\boldsymbol M$. According to the specific structure of the matrix $\boldsymbol M$, one can instantiate the operator matrices $\boldsymbol A$ and $\boldsymbol B$ with desirable properties. Here, for our purposes, we introduce one of the customary choices of $\boldsymbol A$ and $\boldsymbol B$, the unit $f$-circulant matrix $\boldsymbol Z_f$, which we will define next.
\begin{definition}[unit $f$-circulant Matrix]
For a real-valued scalar $f$, an $n\times n$ unit $f$-circulant matrix is defined as follows,
\begin{equation}\label{f-unit-Circulant matrix}
\begin{aligned}
\boldsymbol Z_f=\left(
  \begin{array}{cccc}
    0 &0 & \cdots & f \\
    1 & 0 & \cdots & 0 \\
    \vdots & \vdots & \vdots & \vdots \\
    0 & \cdots & 1 & 0 \\
  \end{array}
\right).
\end{aligned}
\end{equation}
\end{definition}

It is easy to verify that $\boldsymbol Z_1,\boldsymbol Z_{-1}$ are unitary matrices, as well as $\boldsymbol Z_1^i,\boldsymbol Z_{-1}^i,i=0,1,\ldots n-1$. By inverting the displacement operators with operator matrices $(\boldsymbol Z_{1},\boldsymbol Z_{-1})$, we then demonstrate how to decompose an $n\times n$ matrix as linear combinations of unitaries.

\begin{theorem}\label{decomposition theorem}
Let $\boldsymbol M\in \mathbb{C}^{n\times n}$, $m_{i,k}$ be the $k$-th element of the $i$-th row of $\boldsymbol M$, $\begin{small}\boldsymbol J=\left(
         \begin{array}{ccc}
         &  & 1 \\
         & \iddots &  \\
         1 &  &  \\
         \end{array}
         \right)
         \end{small}$
is the reversal matrix, and
\begin{equation}
g(k):=\left\{\begin{array}{cc}
0 &\quad k=0,1,2\cdots n-2, \\
1 & \qquad\qquad k=n-1,\\
\end{array}\right.
\end{equation}
then $\boldsymbol M$ can be decomposed as:\\
i)
\begin{equation}
  \boldsymbol M=\frac{1}{2}\sum_{i,k=0}^{n-1}\hat{m}_{i,k}\boldsymbol Z_{1}^{i}\boldsymbol J\boldsymbol Z_{-1}^{n-1-k},
\end{equation}
where $\hat{m}_{i,k}=m_{i,k}-(-1)^{g(k)}m_{(i-1)\mathrm{mod}\,n,(k+1)\mathrm{mod}\,n}$ is  the $k$-th element of the $i$-th row of matrix $\Delta_{\boldsymbol Z_1,\boldsymbol Z_{-1}}[\boldsymbol M]$.\\
ii)
\begin{equation}\label{slym}
   \boldsymbol M=\frac{1}{2}\sum_{i,k=0}^{n-1}\tilde{m}_{i,k}\boldsymbol Z_{1}^{i}\boldsymbol Z_{-1}^{n-1-k},
\end{equation}
where $\tilde{m}_{i,k}=m_{(i-1)\mathrm{mod}\,n,k}-(-1)^{g(k)}m_{i,(k+1)\mathrm{mod}\,n}$ is the $k$-th element of the $i$-th row of matrix $\nabla_{\boldsymbol Z_1, \boldsymbol Z_{-1}}[\boldsymbol M]$.
\end{theorem}
\proof \: See Appendix A.

We call these two decompositions Stein type and Sylvester type, respectively.  From this theorem, using the displacement matrices $\{\boldsymbol J,\boldsymbol Z_{1}^{i},\boldsymbol Z_{-1}^{i},\,i=0,1,\cdots, n-1\}$ as the elementary components, one can decompose an $n\times n$ matrix into linear combinations of these simple unitaries, and the decomposition coefficients are the elements of the displacement of the matrix which can be easily calculated. The proposed LCU decompositions actually provide a way to parameterize the decomposed matrices, that is, we can use the elements of $\mathcal{L}(\boldsymbol M)$ as a parameterized representation of  $\boldsymbol M$.
\subsection{Implement Block-encodings of the matrices with displacement structures}\label{seciiib}
In this section, we will show that the Toeplitz/Hankel matrices and  their generalizations can generate elegant parameterized representations with nearly $O(n)$ parameters when associated with operator matrices $(\boldsymbol Z_{1},\boldsymbol Z_{-1})$. Furthermore, we will illustrate how to implement the block-encodings of such matrices in detail.

Taking the Toeplitz matrices as an example, we compute their Sylvester displacement first:
\begin{equation}\label{sylofT}
\begin{aligned}
  &\nabla_{\boldsymbol Z_1, \boldsymbol Z_{-1}}[\boldsymbol T_n]\\
                           &=\left(
                               \begin{array}{cccc}
                                 t_{n-1}-t_{-1} & t_{n-2}-t_{-2} & \cdots & 2t_{0}         \\
                                 0              & 0              & \cdots & t_{-(n-1)}+t_{1}\\
                                 \vdots         & \vdots         & \vdots & \vdots         \\
                                 0              & 0              & \cdots & t_{-1}+t_{n-1} \\
                               \end{array}
                             \right).\\
\end{aligned}
\end{equation}
Then $\boldsymbol T_n$ can be decomposed into a linear combination of unitaries as follow,
\begin{small}
\begin{equation}\label{decT}
\begin{aligned}
   \boldsymbol T_n
           &=\frac{1}{2}\big[2t_0 \boldsymbol I+(t_{1}+t_{-(n-1)}) \boldsymbol Z_{1}^{1}+\ldots+(t_{n-1}+t_{-1}) \boldsymbol Z_{1}^{n-1}\\
           &+(t_{1}-t_{-(n-1)}) \boldsymbol Z_{-1}^{1}+\ldots+(t_{n-1}-t_{-1}) \boldsymbol Z_{-1}^{n-1}\big].
\end{aligned}
\end{equation}
\end{small}
We would like to emphasize that since the Toeplitz matrices are represented with $2n-1$ parameters, this decomposition should be optimal on the number of items.

To implement the block-encodings of Toeplitz matrices from Eq.(\ref{decT}), we define two state preparation operators as follows,
\begin{equation}\label{stapreT1}
\begin{aligned}
&\boldsymbol V_{(\nabla[\boldsymbol T_n])}|0\rangle=|V_{(\nabla[ \boldsymbol T_n])}\rangle\\
&=\frac{1}{\sqrt{\chi_{\boldsymbol T_n}}}\sum_{j=0}^{n-1}\sqrt{t_j+t_{-[(n-j)\,\textrm{mod}\,n]}}|j\rangle\\
&+\frac{1}{\sqrt{\chi_{\boldsymbol T_n}}}\sum_{j=n}^{2n-1}\sqrt{t_{(j-n)}-t_{-[(2n-j)\,\textrm{mod}\,n]}}|j\rangle\\
&\equiv\frac{1}{\sqrt{\chi_{\boldsymbol T_n}}}\sum_{j=0}^{2n-1}\sqrt{\tilde{t}_j}|j\rangle,
\end{aligned}
\end{equation}
\begin{equation}\label{stapreT2}
\boldsymbol V_{(\nabla[\boldsymbol T_n]^{\ast})}|0\rangle=|V_{(\nabla[ \boldsymbol T_n]^{\ast})}\rangle=\frac{1}{\sqrt{\chi_{\boldsymbol T_n}}}\sum_{j=0}^{2n-1}\sqrt{\tilde{t}_{j}^{\ast}}|j\rangle,
\end{equation}
where $\chi_{\boldsymbol T_n}=\sum_{j=0}^{2n-1}|\tilde{t}_{j}|$, and the square root operation takes the main square root of $\tilde{t}_{j}$ and $\tilde{t}^{\ast}_{j}$. Then, we define a controlled unitary
\begin{equation}\label{controluT}
\textrm{select} \boldsymbol U_{T_n}=\sum_{j=0}^{n-1}|j\rangle\langle j|\otimes \boldsymbol Z_1^{j}+\sum_{j=n}^{2n-1}|j\rangle\langle j|\otimes \boldsymbol Z_{-1}^{j-n}
\end{equation}
Since $\sqrt{\tilde{t}_{j}}(\sqrt{\tilde{t}^{\ast}_{j}})^\ast=\tilde{t}_{j}$, it is easy to verify that
\begin{equation}
\boldsymbol T_n=\frac{\chi_{\boldsymbol T_n}}{2}(\langle0|( \boldsymbol V^\dagger_{(\nabla [\boldsymbol T_n]^{\ast})}\otimes \boldsymbol I) \textrm{select}\boldsymbol U_{T_n}\boldsymbol (V_{(\nabla [\boldsymbol T_n])}\otimes \boldsymbol I)|0\rangle ),
\end{equation}
which means that $\boldsymbol V^\dagger_{(\nabla [\boldsymbol T_n]^{\ast})} \textrm{select}\boldsymbol U_{T_n}\boldsymbol V_{(\nabla [\boldsymbol T_n])}$ is a block-encoding of $\boldsymbol T_n$.

Implementing the block-encoding of $\boldsymbol T_n$ is to implement $\boldsymbol V_{(\nabla[\boldsymbol T_n])}$, $\boldsymbol V_{(\nabla[\boldsymbol T_n]^{\ast})}$ and $\textrm{select} \boldsymbol U_{T_n}$.  For quantum state preparation operators, we need to construct a quantum circuit that prepare $|V_{(\nabla[ \boldsymbol T_n])}\rangle$ reversibly. In the black-box model, this can not be done by the traditional black-box quantum state preparation \cite{grover2000synthesis} because its success probability can not be increased arbitrarily close to certainty, which will make the final error uncontrollable. We provide a reversible algorithm called steerable black-box quantum state preparation, based on fixed-point amplitude amplification \cite{yoder2014fixed}, which can prepare a quantum state with the success probability increasing arbitrarily close to certainty. Since it will be called multiple times as a key subroutine of our algorithm and may be of independent interest to other quantum algorithms, we make a formal statement here.



\begin{lemma}
For a vector $\boldsymbol x\in \mathbb{C}^{n\times1}$, $|\boldsymbol x|_{\textrm{max}}=\max_{i}|\sqrt{x_i}|$ is a small constant, and the elements are given by a black box $\mathcal{O}_{x}$  acting as
\begin{equation}
  \mathcal{O}_x|i\rangle|0\rangle\rightarrow|i\rangle|x_i\rangle.
\end{equation}
Then, the steerable black-box quantum state preparation algorithm generate a state, up to a global phase, that is $\epsilon_p$-approximation of
\begin{equation}
  |x\rangle=\frac{1}{\sqrt{\|\boldsymbol x\|_1}}\sum_{i=0}^{n-1}\sqrt{x_i}|i\rangle
\end{equation}
with success probability at least $1-\delta^2$, using $O\Big(\frac{\sqrt{n}\mathrm{log}(1/\delta)}{\sqrt{\|\boldsymbol x\|_1}}\Big)$ queries of $\mathcal{O}_x$ and additional $O\Big(\frac{\sqrt{n}\mathrm{log}(1/\delta)}{\sqrt{\|\boldsymbol x\|_1}}\mathrm{polylog}(\frac{n}{\epsilon_p\sqrt{\|\boldsymbol x\|_1}})\Big)$ elementary gates.
\end{lemma}
\proof \: See Appendix B.

The operators defined by Eqs. (\ref{stapreT1}), (\ref{stapreT2}) and (\ref{controluT}) are actually the operators of LCU circuit \cite{Childs2017Quantum} which has been used in many quantum algorithms \cite{chakraborty2019power,van2017quantum,berry2015simulating, gui2009allowable}. Here, we implement this circuit with two different data access models introduced in Sec.\ref{seciic}.  The method in the QRAM data structure model is especially useful for the structured matrices whose displacements have been stored in the data structure. And the method in the black-box model will have a wider range of applications because of the flexibility of its implementation. We summarize the results as follows.



\begin{theorem}\label{the2}
Let $\boldsymbol T_n\in \mathbb{C}^{n\times n}$ be a Toeplitz matrix. (i) If the elements of $\boldsymbol T_n$ are provided by a black box  $\mathcal{O}_{\boldsymbol T_n}$, i.e.,
\begin{equation*}
  \mathcal{O}_{\boldsymbol T_n}|i\rangle|k\rangle|0\rangle=|i\rangle|k\rangle|t_{i,k}\rangle,
\end{equation*}
one can implement a $(\chi_{\boldsymbol T_n}/2;\mathrm{log}(n)+2;\epsilon)$-block-encoding of $\boldsymbol T_n$ with $O\Big(\frac{\sqrt{n}\mathrm{log}(\chi_{\boldsymbol T_n}/\epsilon)}{\sqrt{\chi_{\boldsymbol T_n}}}\Big)$ uses of $\mathcal{O}_{\boldsymbol T_n}$ and additionally using $O\Big(\frac{\sqrt{n}}{\sqrt{\chi_{\boldsymbol T_n}}}\mathrm{polylog}(\frac{n\chi_{\boldsymbol T_n}}{\epsilon})\Big)$ elementary gates. (ii) If the nonzero elements of Sylvester displacement of $\boldsymbol T_n$, i.e., $\{\tilde{t}_j\}_{j=0}^{2n-1}$ are stored in the QRAM data structure as shown in lemma \ref{datastructureT}, one can  implement a $(\chi_{\boldsymbol T_n}/2;\mathrm{log}(n)+1;\epsilon)$-block-encoding of $\boldsymbol T_n$ with gate complexity $O(\mathrm{polylog}(n\chi_{\boldsymbol T_n}/\epsilon))$ and memory cost $O(n)$.
\end{theorem}
\proof \: See Appendix C.

Moreover, we show that how to implement the block-encoding of the Toeplitz-like matrices. Let $\boldsymbol T_L\in \mathbb{C}^{n\times n}$ be a Toeplitz-like matrix, $\boldsymbol (T_L)_{i,k}=\tau_{i,k}$, $(\nabla_{\boldsymbol Z_1, \boldsymbol Z_{-1}}[\boldsymbol T_L])_{i,k}=\tilde{\tau}_{i,k}$. We first define two state preparation operators as follows,

\begin{equation}\label{sltl}
\boldsymbol V_{(\nabla[\boldsymbol T_L])}|0\rangle|0\rangle=|V_{(\nabla[ \boldsymbol T_L])}\rangle=\frac{1}{\sqrt{\chi_{\boldsymbol T_L}}}\sum_{i=0}^{n-1}\sum_{k=0}^{n-1}\sqrt{\tilde{\tau}_{i,k}}|i\rangle|k\rangle,
\end{equation}
\begin{equation}\label{sltl2}
\boldsymbol V_{(\nabla[\boldsymbol T_L]^{\ast})}|0\rangle|0\rangle=|V_{(\nabla[ \boldsymbol T_L]^{\ast})}\rangle=\frac{1}{\sqrt{\chi_{\boldsymbol T_L}}}\sum_{i=0}^{n-1}\sum_{k=0}^{n-1}\sqrt{\tilde{\tau}_{i,k}^{\ast}}|i\rangle|k\rangle,
\end{equation}
where $\chi_{\boldsymbol T_L}=\sum_{i=0}^{n-1}\sum_{k=0}^{n-1}|\tilde{\tau}_{i,k}|$, and the square root operation takes the main square root of $\tilde{\tau}_{i,k}$ and $\tilde{\tau}^{\ast}_{i,k}$. Then, we define
\begin{equation}\label{cutl}
 \textrm{select} \boldsymbol U_{T_L}=(\sum_{i=0}^{n-1}|i\rangle\langle i|\otimes \boldsymbol I\otimes \boldsymbol Z_{1}^{i})(\sum_{k=0}^{n-1}\boldsymbol I\otimes |k\rangle\langle k|\otimes \boldsymbol Z_{-1}^{n-1-k}).
\end{equation}
Since $\sqrt{\tilde{\tau}_{i,k}}(\sqrt{\tilde{\tau}^{\ast}_{i,k}})^\ast=\tilde{\tau}_{i,k}$, it is easy to verify that
\begin{equation}
\boldsymbol T_L=\frac{\chi_{\boldsymbol T_L}}{2}(\langle0| (\boldsymbol V^\dagger_{(\nabla [\boldsymbol T_L]^{\ast})}\otimes \boldsymbol I) \textrm{select}\boldsymbol U_{T_L}(\boldsymbol V_{(\nabla [\boldsymbol T_L])}\otimes \boldsymbol I)|0\rangle ),
\end{equation}
which means that $(\boldsymbol V^\dagger_{(\nabla [\boldsymbol T_L]^{\ast})}\otimes \boldsymbol I) \textrm{select}\boldsymbol U_{T_L}(\boldsymbol V_{(\nabla [\boldsymbol T_L])}\otimes \boldsymbol I)$ is a block-encoding of $\boldsymbol T_L$.

It seems difficult to implement the state preparation operators $\boldsymbol V_{(\nabla[\boldsymbol T_L])}$ and $\boldsymbol V_{(\nabla[\boldsymbol T_L]^{\ast})}$ with complexity less than $O(n)$ in the black-box model. However, as shown in Eq.(\ref{sylofT}), the Sylvester displacement of a Toeplitz matrix has non-zero elements only along its first row and last column. For a Toeplitz-like matrix $\boldsymbol T_L$, there are a few elements on some diagonals of the matrix that are not equal to the others. Then, it can be directly verified that the sub-matrix left by deleting the first row and last column of the Sylvester displacement of the Toeplitz-like matrix is a sparse matrix.  Based on this observation, we construct the block-encodings of the Toeplitz-like matrices, and the results are summarized as follows.

\begin{corollary}\label{the3}
Let $\boldsymbol T_L\in \mathbb{C}^{n\times n}$ be a Toeplitz-like matrix. Suppose that the sub-matrix left by deleting the first row and last column of $\nabla_{\boldsymbol Z_1, \boldsymbol Z_{-1}}[\boldsymbol T_L]$ is $(d-1)$-row-sparse, i.e., there are at most $(d-1)$ nonzero elements in each row. (i) If the elements of $\boldsymbol T_L$ are provided by a black box $\mathcal{O}_{\boldsymbol T_L}$, i.e.,
\begin{equation*}
  \mathcal{O}_{\boldsymbol T_L}|i\rangle|k\rangle|0\rangle=|i\rangle|k\rangle|\tau_{i,k}\rangle,
\end{equation*}
and a black box that computes the positions of the distinct elements on diagonals of the Toeplitz-like matrices is provided, one can implement a $(\chi_{\boldsymbol T_L}/2;2\mathrm{log}(n)+2;\epsilon)$-block-encoding of $\boldsymbol T_L$ with $O\Big(\frac{\sqrt{nd}\mathrm{log}(\chi_{\boldsymbol T_L}/\epsilon)}{\sqrt{\chi_{\boldsymbol T_L}}}\Big)$ uses of $\mathcal{O}_{\boldsymbol T_L}$ and additionally using $O\Big(\frac{\sqrt{nd}}{\sqrt{\chi_{\boldsymbol T_L}}}\mathrm{polylog}(\frac{nd\chi_{\boldsymbol T_L}}{\epsilon})\Big)$ elementary gates. (ii) If the nonzero elements of Sylvester displacement of $\boldsymbol T_L$, i.e., $\{\tilde{\tau}_{i,k}\}_{i,k=0}^{n-1}$ are stored in the QRAM data structure as shown in lemma \ref{datastructureTL}, one can implement a $(\chi_{\boldsymbol T_L}/2;2\mathrm{log}(n);\epsilon)$-block-encoding of $\boldsymbol T_L$ with gate complexity $O(\mathrm{polylog}(n\chi_{\boldsymbol T_L}/\epsilon))$ and memory cost $O(dn\mathrm{log}\,n)$.
\end{corollary}
\proof \: See Appendix D.

\textbf{Remark 1.} One might be confused about the QRAM data structure used in this paper, which stores $\tilde{m}_{i,k}$ instead of $m_{i,k}$ for a matrix $\boldsymbol M$. In fact, in most of quantum algorithms using this data structure, such as \cite{chakraborty2019power,kerenidis2020quantum}, the stored entries are $m_{i,k}^{p},p\in[0,2]$. Since $\tilde{m}_{i,k}$, as defined below Eq. (\ref{slym}), can be calculated as efficiently as $m_{i,k}^{p},p\in[0,2]$, our assumption about such data structure is not stronger than the assumption in the previous algorithms.

\textbf{Remark 2.} The result (i) and result (ii) in theorem 2, as well as corollary 1, use different data access models, where the black-box model queries the elements of structured matrices, while the QRAM data structure model requires the elements of their displacements have been stored. Due to the different data access models, a direct comparison of the complexity of these results is inappropriate. And it is unwise to claim which result is more advantageous based on the complexity.  In practical applications, one should choose the appropriate method according to the pattern the data can be obtained.

For a circulant matrix $\boldsymbol C_n$, computing its Sylvester displacement, the LCU decomposition of $\boldsymbol C_n$ is $\boldsymbol C_n =\sum_{j=0}^{n-1}c_j\boldsymbol Z_1^j$. Similar to the implementation of the block-encodings of the Toeplitz matrices, we can implement block-encodings of $\boldsymbol C_n$. Since the number of decomposed items of the circulant matrices is less than that of the Toeplitz matrices, the resources required to implement these block-encodings are less than those stated in theorem 2. The same conclusion holds for circulant-like matrices.


For a Hankel matrix $\boldsymbol H_n$, it can be decomposed as follows by computing their Stein displacements,
\begin{equation}
\begin{aligned}
 &\boldsymbol H_{n}=\frac{1}{2}[2h_{n-1}\boldsymbol J+(h_{n}+h_{0})\boldsymbol Z_{1}^{1}\boldsymbol J+\cdots\\ &+(h_{2n-2}+h_{n-2})\boldsymbol Z_{1}^{n-1}\boldsymbol J+(h_{n-2}-h_{2n-2})\boldsymbol J\boldsymbol Z_{-1}^{1}\\&+\cdots+(h_{0}-h_{n})\boldsymbol{JZ}_{-1}^{n-1}].
 \end{aligned}
\end{equation}
Note that $\boldsymbol J \boldsymbol Z_{-1}^{i}=-\boldsymbol Z_{-1}^{n-i}\boldsymbol J$, and this decomposition is equivalent to
\begin{equation}
\begin{aligned}
 &\boldsymbol H_{n}=\frac{1}{2}[2h_{n-1}\boldsymbol J+(h_{n}+h_{0})\boldsymbol Z_{1}^{1}\boldsymbol J+\cdots\\ &+(h_{2n-2}+h_{n-2})\boldsymbol Z_{1}^{n-1}\boldsymbol J+(h_{2n-2}-h_{n-2})\boldsymbol Z_{-1}^{n-1}\boldsymbol J\\&+\cdots+(h_{n}-h_{0})\boldsymbol{Z}_{-1}^{1}\boldsymbol J].
 \end{aligned}
\end{equation}
Since $\boldsymbol J=\boldsymbol\sigma_x^{\otimes \textrm{log}n}$ ($\boldsymbol\sigma_x$ is Pauli-X operator), we can implement an $\epsilon$-approximate block-encoding of $\boldsymbol  H_n$ by constructing a quantum circuit similar to the block-encoding implementation of $\boldsymbol  T_n$, where the scaling factor is $\chi'_{\boldsymbol  H_n}/2=\sum_{i=0}^{n-1}\sum_{k=0}^{n-1}|\hat{h}_{i,k}|/2$. Also, the block-encodings of the Hankel-like matrices can be implemented similarly to that of the Toeplitz-like matrices.

In many cases, such as visual tracking \cite{yu2019quantum}, we need to extend the non-Hermitian matrices with displacement structures to Hermitian. For the extended matrices
\begin{equation}\label{extend}
  \overline{\boldsymbol M}=\left(
                        \begin{array}{cc}
                          \boldsymbol 0& \boldsymbol M \\
                          \boldsymbol M^\dagger &\boldsymbol 0 \\
                        \end{array}
                      \right),
\end{equation}
let $\boldsymbol U$ be a $(\chi;a;\epsilon)$-block-encoding of $\boldsymbol M$, then we can implement a $(\chi;a;\epsilon)$-block-encoding of $\overline{\boldsymbol M}$  by using the method of complementing block-encoded matrices \cite{chakraborty2019power}. The cost of implementing this block-encoding is nearly twice the cost of implementing $\boldsymbol U$. Therefore, without losing generality, we can assume that the matrices studied in the following sections are Hermitian.
\subsection{ Quantum Algorithm for Linear Systems with Displacement Structures}\label{seciiic}
As mentioned, given a block-encoding $\boldsymbol U$ of a matrix $\boldsymbol M$, one can perform a number of useful operations on $\boldsymbol M$. In particular, combining the variable-time amplitude amplification technique \cite{ambainis2012variable} and the idea of implementing smooth functions of block-Hamiltonians \cite{van2017quantum}, Chakraborty et al \cite{chakraborty2019power} presented a quantum algorithm for linear systems  within the block-encoding framework. We invoke the complexity of this algorithm as follows.

\begin{lemma} \label{inverse}(Variable-time quantum linear systems algorithm \cite{chakraborty2019power})
Let $\boldsymbol H$ be an $n\times n$ Hermitian matrix, $\lambda_i$ are the nonzero eigenvalues of $\boldsymbol H$ such that $\lambda_i\in[-1,-1/\kappa_{\boldsymbol H}]\cup[1/\kappa_{\boldsymbol H},1]$, where $\kappa_{\boldsymbol H}>2$ is the condition number of $\boldsymbol H$. Suppose that there is a $(\alpha; a; \delta)$-block-encoding $\boldsymbol U$ of $\boldsymbol H$, where $\delta=o(\epsilon/(\kappa_{\boldsymbol H}^2 \mathrm{log}^3(\frac{\kappa_{\boldsymbol H}}{\epsilon})))$, and $\boldsymbol U$ can be implemented in time $T_U$. Also suppose the state $|b\rangle$ can be prepared in time $T_b$. Then there exists a quantum algorithm that produces a state that is  $\epsilon$-close to $\boldsymbol H^{-1}|b\rangle/\|\boldsymbol H^{-1}|b\rangle\|$ in time
\begin{equation*}
  O(\kappa_{\boldsymbol H}(\alpha(a+T_U)\mathrm{log}^2(\frac{\kappa_{\boldsymbol H}}{\epsilon})+ T_b)\mathrm{log}(\kappa_{\boldsymbol H})).
\end{equation*}
\end{lemma}





As mentioned in Sec.\ref{seciic}, there are some methods that can prepare right-hand-side state $|b\rangle$ in time $O(\textrm{polylog}\,n)$ under certain conditions. Even if such an efficient state preparation procedure can not be provided, the complexity of preparing $|b\rangle$ will not exceed the complexity of implementing block-encodings of structured matrices in the same data access model. More specifically, in the black-box model, the query complexity of preparing an $n$-dimensional quantum state is $O(\sqrt{n})$ \cite{grover2000synthesis}. In the QRAM data structured model, one can prepare the quantum state $|b\rangle$ with complexity $O(\textrm{polylog}\,n)$ \cite{kerenidis2016quantum}. Here, following the assumption of previous quantum algorithms \cite{eq1,Childs2017Quantum}, we neglect the error in producing $|b\rangle$ since this error is independent of the design of the quantum algorithm.


Therefore, according to lemma 2, the method proposed in theorem \ref{the2} can induce a quantum algorithm to solve the structured linear systems with complexity (i) $\widetilde{O}(\kappa_{\boldsymbol H}\sqrt{\chi}\sqrt{n}\textrm{polylog}(1/\epsilon))$ in the black-box model (We use the symbol $\widetilde{O}$ to hide redundant poly-logarithmic factors. And since the elementary gates requirement in the black-box model is larger than the query complexity by logarithmic factors, we will not describe them individually from now on.); (ii) $\widetilde{O}(\kappa_{\boldsymbol H}\chi\textrm{poly}\textrm{log}(n/\epsilon))$ in the QRAM data structure model. Obviously, this algorithm is expected to be efficient for matrices of which the $\chi$ is small.

Many matrices with displacement structures encountered in a diverse range of applications satisfy this criteria. One of the typical examples should be Toeplitz matrices in the Wiener class \cite{CGT,Toe}. This kind of matrices are usually obtained by the discretization of some continuous problems. More specifically, let $C_{2\pi}$ be the set of all $2\pi$-periodic continuous real-valued functions defined on $[0,2\pi]$. Let $\boldsymbol T_n$ be the $n\times n$ Toeplitz matrices of which the elements of every diagonal are given by the Fourier coefficients of a function $f\in C_{2\pi}$, i.e,
\begin{equation}
t_j=\frac{1}{2\pi}\int_{0}^{2\pi}\!\!\!f(\lambda)e^{-ij\lambda}\ud\lambda, \qquad j=0,\pm1,\pm2,\cdots.
\end{equation}
The function $f$ is called the generating function of the sequence of Toeplitz matrices $\boldsymbol T_n(1\leq n<\infty)$. The sequence of Toeplitz matrices $\boldsymbol T_n(1\leq n <\infty)$ of which the sequence $\{t_j\}$ is absolutely summable is said to be in the Wiener class.  That is to say, for Toeplitz matrices in Wiener class, there must be a constant $\rho$, such that
\begin{equation}
  \sum_{j=-\infty}^{\infty}|t_j|<\rho.
\end{equation}

Thus, for Toeplitz matrices in the Wiener class, we have
\begin{equation}
\begin{aligned}
  \chi_{\boldsymbol T_n} &=2|t_0|+|t_{1}+t_{-(n-1)}|+\ldots+|t_{n-1}+t_{-1}|\\
      &+|t_{1}-t_{-(n-1)}|+\ldots +|t_{n-1}-t_{-1}|\\
      &\leq 2\sum_{j=-(n-1)}^{n-1}|t_j|<2\rho.
 \end{aligned}
\end{equation}
The complexity of the quantum algorithm for solving the Toeplitz systems in Wiener class is
(i) $\widetilde{O}(\kappa_{\boldsymbol T_n}\sqrt{n}\textrm{polylog}(1/\epsilon))$ in the black-box model; (ii) $\widetilde{O}(\kappa_{\boldsymbol T_n}\textrm{poly}\textrm{log}(n/\epsilon))$ in the QRAM data structure model. When the Toeplitz matrices are well-conditioned (We call a matrix $\boldsymbol M$ well-conditioned of which $\kappa_{\boldsymbol M}\in O(\textrm{polylog}\,n)$.) and $1/\epsilon\in O(\textrm{poly}\,n)$, the quantum algorithm is (i) quadratically faster than the classical methods in the black-box model; (ii) exponentially faster than the classical methods in the QRAM data structure model.

As of now, some work regarding Toeplitz matrices have been studied in the quantum setting. In 2018, Wan et al \cite{wan2018asymptotic} adopted associated circulant matrices to approximate the Toeplitz matrices in Wiener class and solved the circulant linear systems by accessing the values of the generating function at specific points in parallel. It is an asymptotic quantum algorithm of which the error is related to the dimension of the Toeplitz matrices. Whether there is an exact quantum algorithm that the error is independent of the dimension is raised as an open question in \cite{wan2018asymptotic}. The algorithm suggested in this section gives the answer, and it is more advantageous when rigorous precision is required or the Toeplitz matrices and their associated circulant matrices do not approach quickly as the dimension increases. Additionally, for the cases where no generating function is provided, our algorithm can improve the dependence on the condition number and precision since the complexity of the quantum algorithm proposed in \cite{wan2018asymptotic} has a quadratic dependence on the condition number and a linear dependence on the precision.


Besides the Toeplitz matrices in the Wiener class, for the circulant matrices $\boldsymbol C_n$, it is often the case in practical applications that $c_j$ are nonnegative for all $j$, and the spectral norm $\|\boldsymbol C_n\|=\sum_{j=0}^{n-1}c_j$ of $\boldsymbol C_n$ are constants. Thus, $\chi_{\boldsymbol C_n}$  will be bounded by some constants, and the quantum algorithm based on proposed block-encodings can solve these circulant linear systems with complexity (i) $\widetilde{O}(\kappa_{\boldsymbol C_n}\sqrt{n}\textrm{polylog}(1/\epsilon))$ in the black-box model; (ii) $\widetilde{O}(\kappa_{\boldsymbol C_n}\textrm{poly}\textrm{log}(n/\epsilon))$ in the QRAM data structure model.



For the circulant matrices described above, based on the observation of LCU decomposition of $\boldsymbol C_n$, Zhou et al. \cite{zhou2017efficient} used the method of simulating Hamiltonian with a truncated Taylor series \cite{berry2015simulating} and HHL algorithm \cite{eq1} to solve the associated linear systems. Under the assumption that there is an oracle that can prepare the state $\frac{1}{\sqrt{\chi_{\boldsymbol C_n}}}\sum_{j=0}^{n-1}\sqrt{c_j}|j\rangle$ in time $O(\textrm{polylog}\,n)$, the complexity of the quantum algorithm proposed in \cite{zhou2017efficient} is $\widetilde{O}(\kappa_{\boldsymbol C_n}^2\textrm{polylog}(n)/\epsilon)$. However, it is not always feasible to prepare this initial state with such a running time, especially in the black-box model. In this paper, we show in detail how to implement the preparation of this state in two different data access models. In particular, when using the same data access model, the algorithm proposed in this paper brings us complexity improvement on $\kappa_{\boldsymbol C_n}$ and $1/\epsilon$, which comes from the use of updated technique for Hamiltonian simulation and linear system solver.

There are some Hankel matrices of which $\sum_{j=0}^{\infty}|h_j|$ are convergent, such as
\begin{equation}
 \boldsymbol H_{i,k}=\frac{1}{(i+k+1)!}
\end{equation}
which arise in determining the covariance structure of an iterated Kolmogorov diffusion \cite{habermann2018explicit}. In addition, there are also some Hankel matrices  generated by the discretization of some functions \cite{widom1966hankel}, just as the Toeplitz matirces in the Wiener class. Then, we can implement block-encodings of these Hankel matrices with bounded scaling factors, which will derive a quantum algorithm that solves the Hankel linear systems with significant speedup (same as that for Toeplitz systems in Wiener class) in both data access models.


Immediately, for the Toeplitz/Hankel-like matrices, if they satisfy a constraint similar to one of the above forms, we can then solve the linear systems with such structures efficiently by using the block-encodings constructed in corollary \ref{the3}. For example, finding a greatest common divisor of univariate polynomials involves solutions of Toeplitz-like systems and in some cases, as shown in \cite{schonhage1985quasi}, the elements of the displacement of the coefficient matrices are absolutely summable. We would like to emphasize the method proposed in this paper will result in efficient quantum linear systems solvers for the structured matrices of which the $\chi$ is small, not only for the matrices introduced in this section.

\section{ Application to Time Series Analysis }
Note that the quantum algorithm introduced in last subsection always outputs a state encoding the solution of the linear system in its amplitudes. Reading out all the classical information of the solution is time-consuming. To illustrate that the quantum speedup is practically achievable, we provide a concrete example where the coefficient matrix satisfies the specification and some useful information can be extracted from the output.

More specifically, we apply the quantum algorithm for the Toeplitz systems in Wiener class to solve the linear prediction problem of time series. Predicting the future value of a discrete time stochastic process with a set of past samples of the process is one of the most important problems in time series analysis. For linear prediction, we need to estimate the predicted value by a linear combination of the past samples.

To present the problem clearly, we first introduce some terminology used in signal processing, for details, see \cite{haykin2005adaptive}. Let $u(k)$ be a discrete-time stationary zero-mean complex-valued process. A finite impulse response (FIR) linear filter of order $n$ is of the form
\begin{equation}
  \hat{u}(i)=\sum_{k=1}^{n}w_k^{\ast}u(i-k)
\end{equation}
where $\hat{u}(i)$ is the filter output based on the data $\{u(k)\}^{i-1}_{k=i-n}$ and $\{w_k\}^{n}_{k=1}$ are the impulse responses of the filter. For the situation of linear prediction, the desired response is $u(i)$, representing the actual sample of the input process at time $i$. The difference between the desired response $u(i)$ and the filter output $\hat{u}(i)$ is called the estimation error. To estimate the desired response, we should choose the impulse responses $\{w_k\}^{n}_{k=1}$  by making the estimation error as small as possible in some statistical sense.

According to Wiener filter theory, when the estimation error are optimized in the mean-square-error sense, the impulse responses $\{w_k\}^{n}_{k=1}$ are given by the solution of the linear system
 \begin{equation}\label{WHeq}
   \boldsymbol R\boldsymbol w=\boldsymbol r.
 \end{equation}
Here,
\begin{equation}
\boldsymbol R=\left(
  \begin{array}{cccc}
    r(0) & r(1) & \cdots & r(n-1) \\
    r^{\ast}(1)& r(0) & \cdots & r(n-2) \\
    \vdots & \vdots& \ddots & \vdots \\
    r^{\ast}(n-1) & r^{\ast}(n-2) & \cdots & r(0) \\
  \end{array}
\right),
\end{equation}
\begin{equation}
 \boldsymbol r= \left(
   \begin{array}{c}
      r^{\ast}(1) \\
      r^{\ast}(2) \\
      \vdots\\
      r^{\ast}(n) \\
    \end{array}
  \right),
\end{equation}
where $r(k)=\mathbb{E}[u(j)u^{\ast}(j-k)]$ ($\mathbb{E}$ is the expectation operator) is the autocovariances of the input process for lag $k$. This linear system is commonly called the Wiener-Hopf equations.

Note that the covariance matrix $\boldsymbol R$ is an $n\times n$ Hermitian Toeplitz matrix and is almost always positive definite. For a discrete-time stationary process, if the autocovariances of the process are absolutely summable, i.e., $\sum_{k=-\infty}^{\infty}|r(k)|<\infty$, then the function $\tilde{f}(\lambda)$ that takes $r(k)$ as its Fourier coefficients is called the power spectral density function of the process. The power spectral density functions ordinarily exist for the stochastic processes encountered in the physical sciences and engineering. Thus,  $\boldsymbol R $ is a Toeplitz matrix generated by $\tilde{f}(\lambda)$ and in the Wiener class. Moreover,  the eigenvalues $\lambda_k$ of a Hermitian Toeplitz matrix satisfy
\begin{equation}
f_{min}\leq\lambda_k\leq f_{max},
\end{equation}
where $f_{min},f_{max}$ represent the smallest value and the largest value of  the generating function respectively. When the spectral density function is bounded (that can be guaranteed by the continuity of $\tilde{f}(\lambda)$ on $[0,2\pi]$), the condition number of $\boldsymbol R$ will also be bounded.

For the case of known statistics, i.e., the autocovariances of the stationary process are known, one can the query the elements of the covariance matrix $\boldsymbol R$ by the ``black box". Alternatively, the covariance matrix $\boldsymbol R$ can be stored in the QRAM data structure as shown in lemma \ref{datastructureT} in advance. Similarly, the vector $\boldsymbol r$ can also be provided with two different data access models. Then, we can  prepare quantum state $|r\rangle$ with complexity (i) $O(\sqrt{n}/\|\boldsymbol r\|_2)$ in the black-box model \cite{sanders2019black}; (ii) $O(\textrm{polylog}(n/\epsilon))$ in the QRAM data structure model \cite{kerenidis2016quantum}. Note that $\|\boldsymbol r\|_2=O(\chi_{\boldsymbol R})$ will be a constant. By calling the quantum algorithm for solving the Toeplitz systems, we can get a quantum state $|w\rangle$  proportional to the solution of Eq. (\ref{WHeq}) with complexity (i) $\tilde{O}(\sqrt{n}\textrm{polylog}(1/\epsilon))$ in the black-box model; (ii) $\tilde{O}(\textrm{polylog}(n/\epsilon))$ in the QRAM data structure model.

Given the vector $\boldsymbol u=[u(i-1),\cdots,u(i-n)]^{T}$ with appropriate data access model, the quantum state $|u\rangle$ can also be prepared by the methods described in Sec.\ref{seciiic}. Then, the filter output $\hat{u}(i)=\sum_{k=1}^{n}w_k^{\ast}u(i-k)$ can be approximately computed up to some factor by evaluating the inner product of $|u\rangle$ and $|w\rangle$ using Hadamard test \cite{aharonov2009polynomial,wu2021quantum}. Since there is no need to read out all the value of the obtained quantum state, quantum speedup of solving the linear system is preserved. This process in fact provides an example that quantum algorithms can yield significant speedup for the problems of practical interest.
\section{Discussion}
There are some special cases of matrices with displacement structures that can be decomposed into linear combinations of displacement matrices with a few items. The simplest case is banded Toeplitz matrices, $\bar{\boldsymbol T}_n$, of which $t_k =0, |k|> \varrho$ for a constant $\varrho$.  The linear systems of banded Toeplitz matrices occur in many applications, involving the numerical solution of certain differential equations, the modeling of queueing problems, digital filtering, and so on.

Computing the Sylvester displacements of banded Toeplitz matrices, they can be decomposed as follows,
\begin{equation}\label{bandedtd}
\begin{aligned}
   \bar{\boldsymbol T}_n
           &=\frac{1}{2}\big[2t_0 \boldsymbol I+t_{1} \boldsymbol Z_{1}^{1}+\ldots+t_{\varrho} \boldsymbol Z_{1}^{\varrho}\\
           &+t_{-\varrho} \boldsymbol Z_{1}^{n-\varrho}+\ldots+t_{-1}\boldsymbol Z_{1}^{n-1}\\
           &+t_{1} \boldsymbol Z_{-1}^{1}+\ldots+t_{\varrho} \boldsymbol Z_{-1}^{\varrho}\\
           &+(-t_{-\varrho}) \boldsymbol Z_{-1}^{n-\varrho}+\ldots+(-t_{-1})\boldsymbol Z_{-1}^{n-1}].
\end{aligned}
\end{equation}
Similarly to Eqs. (\ref{stapreT1}), (\ref{stapreT2}), and (\ref{controluT}), we can define two state preparation operators $\boldsymbol V_{(\nabla[ \bar{\boldsymbol T}])},  \boldsymbol V_{(\nabla[ \bar{\boldsymbol T}]^\ast)}$ and a controlled unitary operator $\textrm{select}(\boldsymbol U_{ \bar{\boldsymbol T}})$. Then, the unitaries $\boldsymbol V_{(\nabla[ \bar{\boldsymbol T}])}$ and $\boldsymbol V_{(\nabla[\bar{\boldsymbol T}]^\ast)}$ can be implemented by using the generic state preparation algorithm described in \cite{shende2006synthesis}, which requires a gate cost of $O(\varrho)$. Also, the controlled unitary operator $\textrm{select}(\boldsymbol U_{ \bar{\boldsymbol T}})$ can be implemented with $O(\varrho \textrm{log}\,n)$ primitive gates, as the quantum circuit of each $\boldsymbol Z_{1}^{j} (\boldsymbol Z_{-1}^{j})$ only requires $O(\textrm{log}\,n)$ primitive gates.  Thus, we can efficiently implement a $(\chi_{\bar{\boldsymbol T}_n}/2;\lceil\textrm{log} (4\varrho+1)\rceil;\epsilon)$-block-encoding of $\bar{\boldsymbol T}_n$, where $\chi_{\bar{\boldsymbol T}_n}=\sum_{j=-\varrho}^{\varrho}|t_j|.$ It should be noted that the implementation scheme proposed here may provide more facilitation when constructing practical circuits of the block-encodings of such matrices since it does not require any oracle or QRAM. Combined with quantum algorithm for linear systems within the block-encoding framework, it can offer an exponential improvement in the dimension of the linear systems over classical methods.

In the black-box model, one can use the method of \cite{low2019hamiltonian} to efficiently implement a block-encoding for a sparse matrix, where the scaling factor linearly depends on the sparsity (the maximum number of nonzero entries in any row or column) of the matrix. Since the structured matrices studied in this paper are not sparse, the scaling factor of this block-encoding will be $O(n)$. Then, the quantum algorithm for structured linear systems based on such block-encoding can not provide speedup compared with the classical algorithm. We note that \cite{gilyen2019quantum} provided a way to improve the scaling factor. However it require that the upper bound on the the $p$-norm of the rows of the matrix is known, which is not the circumstances considered in this paper.

In the QRAM data structure model, the method stated in \cite{chakraborty2019power,kerenidis2020quantum} also implements block-encodings based on the assumption that the powers of the elements of a matrix are stored in the quantum-accessible data structure beforehand. For a matrix with a displacement structure, the scaling factor produced by the method of \cite{chakraborty2019power,kerenidis2020quantum} can be in the same order of magnitude as that in this paper. However, it should be noted that constructing the data structure that stores some entries about the matrices may constitute the main restriction of this data access model. With respect to this, our method is more advantageous. More specifically, the method of \cite{chakraborty2019power,kerenidis2020quantum} would require a QRAM with data structure storing $O(n^2)$ entries for the matrices with displacement structures. In our method, since we represent these structured matrices with $O(n)$ entries of their displacements, a QRAM with data structure storing $O(n)$ entries is required. Obviously, the data structure in our method can be constructed more rapidly and uses less memory space, so that our method will be more favorable when solving problems involving matrices with displacement structures.


As analyzed above, the origin of the advantages of our algorithm is the succinct parameterized representations of the matrices with displacement structures which is attributed to the Stein and Sylvester type LCU decompositions. There are also some intuitive methods to perform an LCU decomposition. Typically, we specify a set of unitaries that are easy to implement as the basis, and then calculate the decomposition coefficients by solving a linear system with $n^2$ unknown parameters. Alternatively, one can decompose the matrix into a sum of tensor products of Pauli operators, while the number of decomposition items will be considerably larger than ours. These decompositions are high-complexity to construct, and may not even fit to implement block-encodings. In general, it is not easy to find a desirable decomposition.

In particular, in this paper, by the proposed LCU decompositions, we make the implementation of block-encodings of matrices with displacement structure closely related to the task of preparing $O(n)$-dimensional quantum states. Since there are $\Omega(n)$ parameters when defining an $n\times n$ matrix with displacement structure, the query complexity in the black-box model and the memory cost in the QRAM data structure model are nearly optimal. It is an open question whether one can bypass the preparation of $n$-dimensional quantum states and implement block-encodings of the matrices with displacement structures with fewer resources.

It should be noted that although the proposed decompositions in this paper are also available for general dense matrices, the quantum algorithms based on these decompositions can not bring significant improvements over the known results due to the number of decomposition items of $n^2$.  For the same reason, even if other unitaries are used as the basis, a universal LCU decomposition generally cannot give rise to quantum algorithms that surpass existing methods for all dense matrices. Nevertheless, it is still worth exploring specialized decomposition for some specific structured matrices to implement favorable block-encodings, which can result in a significant reduction in the complexity of the quantum linear system solver.

Following Tang's breakthrough work \cite{tang2019quantum}, there is a large class of classical algorithms of which the running time is poly-logarithmic in the dimension. However, this type of speedup is only achievable for the matrices involved to be of low rank. It is not applicable to use the dequantization method to solve linear systems with displacement structures since these matrices are not low-rank in general. Besides, these dequantized algorithms have higher overhead than the quantum algorithms in practice due to their large polynomial dependence on the rank and the other parameters. For the computation of displacement structured matrices, it is an interesting open problem to explore classical algorithms with similar overheads to quantum algorithms.
\section{Conclusion}
In this paper, we demonstrated that several important classes of matrices with displacement structures can be represented and treated similarly by decomposing them into linear combinations of displacement matrices. Based on the devised decompositions, we implemented block-encodings of these structured matrices in two different data access models, and introduced efficient quantum algorithms for solving the linear system with such structures.  The obtained quantum linear systems solvers improved the known results and also motivated some new instances, see Table 1 for a brief summary. In particular, we provided a concrete example to illustrate these quantum algorithms can be used to solve the problems of practical interest with significant speedup.

\begin{widetext}
\begin{center}
\textrm{\footnotesize Table 1: Summary of quantum algorithms for solving linear systems with displacement structures\\[3pt]}
\resizebox{150mm}{!}{
\begin{adjustbox}{center}
\begin{tabular}{|c|c|c|c|}
  \hline
  Coefficient Matrix  & Algorithm  &Remark & Comparison of complexity\\
  \hline\rule{0pt}{17pt}
  \multirow{2}*[-2pt]{\makecell[c]{Toeplitz  matrices \\in the Wiener class }} &Wan et al. \cite{wan2018asymptotic}  & Asymptotic algorithm&\multirow{2}*[4pt]{\makecell[c]{Improve the dependence on the \\condition number and precision,\\ when generating function is unknown}}\\
  \cline{2-3}\rule{0pt}{17pt}
  & Theorem 2-based&Non-asymptotic algorithm &\\
  \hline\rule{0pt}{17pt}
  \multirow{2}*[-2pt]{\makecell[c]{Circulant matrices with \\bounded spectral norm}}& Zhou and Wang \cite{zhou2017efficient}      &\makecell[c]{An oracle of preparing\\ $n$-dimensional state is required}&\multirow{2}*[4pt]{\makecell[c]{Improve the dependence on the \\condition number and precision, \\when using the same data access model}} \\
  \cline{2-3}\rule{0pt}{13pt}
  & Theorem 2-based & Two different data access model&\\
   \hline\rule{0pt}{17pt}
  \multirow{2}*[0pt]{\makecell[c]{Discretized Laplacian\\ (specific banded \\Toeplitz matrices)}}   & Cao et al. \cite{pos}& \makecell[c]{Finite difference discretization\\ of the Poisson equation}& \multirow{2}*[-5pt]{The same order of magnitude}\\
  \cline{2-3}\rule{0pt}{13pt}
  & Discussion-based & \makecell[c]{Applicable to general \\banded Toeplitz matrices}&\\
  \hline\rule{0pt}{17pt}
  \multirow{2}*[-2pt]{\makecell[c]{Toeplitz/Hankel-like \\matrices with small $\chi$}} & \makecell[c]{ No specialized \\quantum algorithm} &$\backslash$&\multirow{2}*[0pt]{$\backslash$}\\
  \cline{2-3}\rule{0pt}{17pt}
  & Corollary 1-based &$\backslash$&\\
  \hline
\end{tabular}
\end{adjustbox}}
\end{center}
\end{widetext}

The presented methods can actually be extended to solve many important computational problems having ties to the structured matrices studied in this paper, such as structured least squares problems and computation of the structured matrices exponential. Also, we hope our work can inspire the study of matrices with displacement structures on near-term quantum devices, since we have shown that these matrices can be decomposed into some easy-to-implement unitaries. At last, there are many other structured matrices such as Cauchy, Bezout, Vandermonde, Loewner, and Pick matrices which are widely employed in various areas.  Designing quantum algorithms for these structured matrices with a dramatic computational acceleration and a major memory-space decrease is worthy of further study.

\section*{Acknowledgements}
The authors would like to thank Yi-Jie Shi for her constructive comments on an early version of this manuscript. This work is supported by the Fundamental Research Funds for the Central Universities (Grant No. 2019XD-A01), the National Natural Science Foundation of China (Grant Nos.  61972048, 61976024, 62006105), the Jiangxi Provincial Natural Science Foundation (Grant No. 20202BABL212004), and the China Scholarship Council (Grant No. 201806470057).


\appendix
\section{Proof of the theorem \ref{decomposition theorem}}

In this appendix, we prove the conclusion in theorem \ref{decomposition theorem}. There are some well-known fundamental results, and the proof of these results can be found in \cite{pan2001structured}. For completeness, we restate here.

\begin{lemma}(\cite{pan2001structured})\label{induction}
For matrices $\boldsymbol A,\boldsymbol B,\boldsymbol M\in \mathbb{C}^{n\times n}$, and $k\geq1$, we have
\begin{equation}
    \boldsymbol M=\boldsymbol A^k\boldsymbol M\boldsymbol B^k+\sum_{i=0}^{k-1}\boldsymbol A^i\Delta_{\boldsymbol A,\boldsymbol B}(\boldsymbol M)\boldsymbol B^i.
\end{equation}
\end{lemma}
\begin{proof}
It is trivial when $k=1$. We assume the identity is true for $k$. Then, multiplying the identity on the left by $\boldsymbol A$ and right by $\boldsymbol B$,
\begin{small}
\begin{equation}
\begin{aligned}
  \boldsymbol{AMB}&=\boldsymbol A^{k+1}\boldsymbol{MB}^{k+1}+\sum_{i=0}^{k-1}\boldsymbol A^{i+1}\Delta_{\boldsymbol A,\boldsymbol B}(\boldsymbol M)\boldsymbol B^{i+1}\\
     &=\boldsymbol A^{k+1}\boldsymbol {MB}^{k+1}+\sum_{i=0}^{k}\boldsymbol A^{i}\Delta_{\boldsymbol A,\boldsymbol B}(\boldsymbol M)\boldsymbol B^{i}-\Delta_{\boldsymbol A,\boldsymbol B}(\boldsymbol M),\\
     \boldsymbol M&=\boldsymbol A^{k+1}\boldsymbol{MB}^{k+1}+\sum_{i=0}^{k}\boldsymbol A^{i}\Delta_{\boldsymbol A,\boldsymbol B}(\boldsymbol M)\boldsymbol B^{i}.
\end{aligned}
\end{equation}
\end{small}
Thus, the identity is true for $k+1$. According to mathematical induction, it is true for all natural numbers.
\end{proof}
\begin{lemma}(\cite{pan2001structured})\label{N decomposition}
If $\boldsymbol A$ is an $a$-potent matrix of order $n$ and $\boldsymbol B$ is a $b$-potent matrix of order $n$, i.e., $\boldsymbol A^n = a\boldsymbol I$ and $\boldsymbol B^n = b\boldsymbol I$, then
\begin{equation}
  \boldsymbol M=\frac{1}{1-ab}\sum_{i=0}^{n-1}\boldsymbol A^i\Delta_{\boldsymbol A,\boldsymbol B}(\boldsymbol M)\boldsymbol B^i.
\end{equation}
\end{lemma}
\begin{proof}
This conclusion is a direct inference of lemma \ref{induction}, when $k=n$ and $\boldsymbol A^n = a\boldsymbol I$, $\boldsymbol B^n = b\boldsymbol I$.
\end{proof}
\begin{lemma}(\cite{pan2001structured})\label{switch}
$\nabla_{\boldsymbol A,\boldsymbol B} = \boldsymbol A\Delta_{\boldsymbol A^{-1},\boldsymbol B}$, if the operator matrix $\boldsymbol A$ is nonsingular, and
$\nabla_{\boldsymbol A,\boldsymbol B} = -\Delta_{\boldsymbol A,\boldsymbol B^{-1}}\boldsymbol B$, if the operator matrix $\boldsymbol B$ is nonsingular.
\end{lemma}
\begin{proof}
Note that if $\boldsymbol A$ is nonsingular, then $\boldsymbol{AM}-\boldsymbol{MB}=\boldsymbol A(\boldsymbol M-\boldsymbol A^{-1}\boldsymbol{MB})$; if $\boldsymbol B$ is nonsingular, then $\boldsymbol{AM}-\boldsymbol{MB}=-\boldsymbol (\boldsymbol M-\boldsymbol A\boldsymbol{MB}^{-1})\boldsymbol B$.
\end{proof}

\begin{definition}[$f$-circulant Matrix]
The $f$-circulant matrix, $\boldsymbol Z_f(\boldsymbol v)$, generated by a unit $f$-circulant matrix and a given vector $\boldsymbol v=[v_0,\cdots,v_{n-1}]^T$ is defined as follows:
\begin{equation}\label{f-Circulant Matrix}
\begin{aligned}
\boldsymbol Z_f(\boldsymbol v)&=(\boldsymbol v \;\;\boldsymbol Z_f\boldsymbol v \;\;\boldsymbol Z_f^2\boldsymbol v \cdots \boldsymbol Z_f^{n-1}\boldsymbol v)\\&=
\left( \begin{array}{cccc}
v_0     & fv_{n-1} & \cdots &fv_{1} \\
v_1     & v_0      & \cdots &  fv_{2}\\
\vdots  & \vdots   & \vdots & fv_{n-1} \\
v_{n-1} & \ldots   & v_1    & v_0 \\
\end{array}\right).
\end{aligned}
\end{equation}
\end{definition}

It turns out that a matrix $\boldsymbol M$ can be expressed as the sums of the products of $f$-circulant matrices and reversal matrix, by inverting the displacement operators.

\begin{theorem}(\cite{pan2001structured})\label{functionexpression}
If a matrix $\boldsymbol M\in\mathbb{C}^{n\times n}$  satisfies $\mathcal{L}(\boldsymbol M)=\boldsymbol{GH}^T$
where $\boldsymbol G=[\boldsymbol g_1\ldots \boldsymbol g_r],\boldsymbol H=[\boldsymbol h_1\ldots \boldsymbol h_r]\in \mathbb{C}^{n\times r}$, $e,f$ are constants, then $\boldsymbol M$ can be expressed as:\\
i)
\begin{equation}
  \boldsymbol M=\frac{1}{1-ef}\sum_{j=1}^{r}\boldsymbol Z_e(\boldsymbol g_j)\boldsymbol Z_{f}(\boldsymbol J\boldsymbol h_j)^T\boldsymbol J,
\end{equation}
where $\mathcal{L}(\boldsymbol M)=\Delta_{\boldsymbol Z_e,\boldsymbol Z_{f}}[\boldsymbol M]$, and $ef\neq1$.\\
ii)
\begin{equation}
   \boldsymbol M=\frac{1}{e-f}\sum_{j=1}^{r}\boldsymbol Z_e(\boldsymbol  g_j)\boldsymbol Z_{f}(\boldsymbol J\boldsymbol h_j),
\end{equation}
where $\mathcal{L}(\boldsymbol M)=\nabla_{\boldsymbol Z_e,\boldsymbol Z_{f}}[\boldsymbol M]$, and $e\neq f$.
\end{theorem}

\begin{proof}
From lemma \ref{N decomposition}, let $\boldsymbol A= \boldsymbol Z_e,\boldsymbol B=\boldsymbol Z_{f}, ef\neq1,$ then we have
\begin{small}
\begin{equation}
\begin{aligned}
   \boldsymbol M&=\frac{1}{1-ef}\sum_{i=0}^{n-1}\boldsymbol Z_e^i\Delta_{\boldsymbol Z_e,\boldsymbol Z_f}(\boldsymbol M)\boldsymbol Z_f^i\\
   &=\frac{1}{1-ef}\sum_{j=1}^{r}\sum_{i=0}^{n-1}\boldsymbol Z_e^i\boldsymbol g_j \boldsymbol h_j^T\boldsymbol Z_f^i\\
   &=\frac{1}{1-ef}\sum_{j=1}^{r}(\boldsymbol g_j \boldsymbol h_j^T+\boldsymbol Z_e\boldsymbol g_j \boldsymbol h_j^T\boldsymbol Z_f+\boldsymbol Z_e^2\boldsymbol g_j \boldsymbol h_j^T\boldsymbol Z_f^2+\cdots\\
   &+\boldsymbol Z_e^{n-1}\boldsymbol g_j \boldsymbol h_j^T\boldsymbol Z_f^{n-1})\\
   &=\frac{1}{1-ef}\sum_{j=1}^{r}[\boldsymbol g_j \;\;\boldsymbol Z_e\boldsymbol g_j\;\;\boldsymbol Z_e^2\boldsymbol g_j\cdots \boldsymbol Z_e^{n-1}\boldsymbol g_j]\\
   &\cdot[\boldsymbol h_j\;\; \boldsymbol Z_{f}^T\boldsymbol h_j\;\; {(\boldsymbol Z_{f}^T)}^2\boldsymbol h_j\cdots {(\boldsymbol Z_{f}^T)}^{n-1}\boldsymbol h_j]^T\\
   &=\frac{1}{1-ef}\sum_{j=1}^{r}\boldsymbol Z_e(\boldsymbol g_j)\\
   &\cdot[\boldsymbol{JJh}_j\;\; \boldsymbol{JZ}_{f}\boldsymbol{Jh}_j\;\; \boldsymbol J{(\boldsymbol Z_{f})}^2 \boldsymbol J \boldsymbol h_j\cdots \boldsymbol J{(\boldsymbol Z_{f})}^{n-1}\boldsymbol{Jh}_j]^T\\
   &=\frac{1}{1-ef}\sum_{j=1}^{r}\boldsymbol Z_e(\boldsymbol g_j) [\boldsymbol J\cdot \boldsymbol Z_{f}(\boldsymbol{Jh}_j)]^T\\
    &=\frac{1}{1-ef}\sum_{j=1}^{r}\boldsymbol Z_e(\boldsymbol g_j)\boldsymbol Z_{f}(\boldsymbol{Jh}_j)^T \boldsymbol J\\
\end{aligned}
\end{equation}
\end{small}
by using the facts $\boldsymbol J^2=\boldsymbol I$ and $\boldsymbol Z_{f}=\boldsymbol J\boldsymbol Z_{f}^{T}\boldsymbol J$.

Furthermore, according to the lemma \ref{switch}, $\Delta_{\boldsymbol Z_{1/e}^{T},\boldsymbol Z_{f}}[\boldsymbol M]=\boldsymbol Z_{1/e}^{T}\nabla_{\boldsymbol Z_e,\boldsymbol Z_{f}}[\boldsymbol M]$, where $\boldsymbol Z_{f}^{-1}=\boldsymbol Z_{1/f}^{T}$, then we can deduce the conclusion for Sylvester type.
\end{proof}

Now, we demonstrate how to decompose an $n\times n$ matrix into linear combinations of unitaries. For our purposes, we choose $(\boldsymbol Z_{1},\boldsymbol Z_{-1})$ as the operator matrices. Note that
~\\
\begin{equation}\label{1}
\begin{aligned}
\boldsymbol Z_{1}(\boldsymbol g_{j})&=
\left(\begin{array}{cccc}
g_{j}^{0}     & g_{j}^{n-1}    & \cdots       & g_{j}^{1}     \\
g_{j}^{1}     & g_{j}^{0}      & \cdots       & g_{j}^{2}     \\
\vdots        & \ddots         & \ddots       & \vdots        \\
g_{j}^{n-1}   & \ldots         & g_{j}^{1}    & g_{j}^{0}     \\
\end{array}\right)\\
&= g_{j}^{0}\boldsymbol Z_{1}^{0}+ g_{j}^{n-1}\boldsymbol Z_{1}^{n-1}+\cdots+g_{j}^{1}\boldsymbol Z_{1}^{1},
\end{aligned}
\end{equation}
~\\
~\\
\begin{equation}\label{2}
\begin{aligned}
\boldsymbol Z_{-1}(\boldsymbol h_{j})&=
\left(\begin{array}{cccc}
h_{j}^{0}     & -h_{j}^{n-1}   & \cdots       & -h_{j}^{1}     \\
h_{j}^{1}     & h_{j}^{0}      & \cdots       & -h_{j}^{2}     \\
\vdots        & \ddots         & \ddots       & \vdots   \\
h_{j}^{n-1}   & \ldots         & h_{j}^{1}    & h_{j}^{0}     \\
\end{array}\right)\\
&= h_{j}^{0}\boldsymbol Z_{-1}^{0}+h_{j}^{n-1}\boldsymbol Z_{-1}^{n-1}+\cdots+h_{j}^{1}\boldsymbol Z_{-1}^{1},
\end{aligned}
\end{equation}
where
\begin{equation}
  \boldsymbol g_j=(g_{j}^{0},g_{j}^{1},\ldots,g_{j}^{n-1})^T,\quad \boldsymbol h_j=(h_{j}^{0},h_{j}^{1},\ldots,h_{j}^{n-1})^T.
\end{equation}

Then, on the one hand
\begin{widetext}
\begin{equation}\label{m1}
\begin{aligned}
  \boldsymbol M&=\frac{1}{2}\sum_{j=1}^{r}\boldsymbol Z_1(\boldsymbol g_j)\boldsymbol Z_{-1}(\mathbf J \boldsymbol h_j)\\
           &=\frac{1}{2}\sum_{j=1}^{r}\big(g_{j}^{0}\boldsymbol Z_{1}^{0}+g_{j}^{1}\boldsymbol Z_{1}^{1}+\cdots+g_{j}^{n-1}\boldsymbol Z_{1}^{n-1}\big)
           \big(h_{j}^{n-1}\boldsymbol Z_{-1}^{0}+h_{j}^{n-2}\boldsymbol Z_{-1}^{1}+\cdots+h_{j}^{0}\boldsymbol Z_{-1}^{n-1}\big)\\
  &=\frac{1}{2}\sum_{j=1}^{r}\big(g_{j}^{0}h_{j}^{n-1}\boldsymbol Z_{1}^{0}\boldsymbol Z_{-1}^{0}+g_{j}^{0}h_{j}^{n-2}\boldsymbol Z_{1}^{0}\boldsymbol Z_{-1}^{1}+\cdots+g_{j}^{0}h_{j}^{0}\boldsymbol Z_{1}^{0}\boldsymbol Z_{-1}^{n-1}\\
            &\quad+g_{j}^{1}h_{j}^{n-1}\boldsymbol Z_{1}^{1}\boldsymbol Z_{-1}^{0}+g_{j}^{1}h_{j}^{n-2}\boldsymbol Z_{1}^{1}\boldsymbol Z_{-1}^{1}+\cdots+g_{j}^{1}h_{j}^{0}\boldsymbol Z_{1}^{1}\boldsymbol Z_{-1}^{n-1}\\
           &\quad+\cdots+\cdots\\
            &\quad+g_{j}^{n-1}h_{j}^{n-1}\boldsymbol Z_{1}^{n-1}\boldsymbol Z_{-1}^{0}+g_{j}^{n-1}h_{j}^{n-2}\boldsymbol Z_{1}^{n-1}\boldsymbol Z_{-1}^{1}+\cdots+g_{j}^{n-1}h_{j}^{0}\boldsymbol Z_{1}^{n-1}\boldsymbol Z_{-1}^{n-1}\big)\\
           &=\frac{1}{2}\sum_{j=1}^{r} \sum_{i,k=0}^{n-1}g_{j}^{i}h_{j}^{k}\boldsymbol Z_{1}^{i}\boldsymbol Z_{-1}^{n-1-k}\\
           &=\frac{1}{2}\sum_{i,k=0}^{n-1}\sum_{j=1}^{r} g_{j}^{i}h_{j}^{k}\boldsymbol Z_{1}^{i}\boldsymbol Z_{-1}^{n-1-k}.
\end{aligned}
\end{equation}
\end{widetext}
On the other hand, since
\begin{equation}
  \nabla_{\boldsymbol Z_1,\boldsymbol Z_{-1}}[\boldsymbol M]=\boldsymbol G\boldsymbol H^T=\sum_{j=1}^{r}\boldsymbol g_j\boldsymbol h_j^T,
\end{equation}
it is immediately verified that
\begin{equation}
  \tilde{m}_{i,k}=\sum_{j=1}^{r} g_{j}^{i}h_{j}^{k},  \qquad i,k=0,1,\cdots,n-1,
\end{equation}
where $\tilde{m}_{i,k}$ is  the $k$-th element of the $i$-th row of matrix $\nabla_{\boldsymbol Z_1,\boldsymbol Z_{-1}}[\boldsymbol M]$.

Therefore,
\begin{equation}
  \boldsymbol M=\frac{1}{2}\sum_{i,k=0}^{n-1}\tilde{m}_{i,k}\boldsymbol Z_{1}^{i}\boldsymbol Z_{-1}^{n-1-k}.
\end{equation}

The decomposition for Stein type can be proved in the same way.

\section{Steerable black-box quantum state preparation}

The scenario for steerable black-box state preparation is as follows. For a vector $\boldsymbol x=(x_0,x_1,\cdots,x_{n-1})^T$, we are provided a black box that returns target elements, i.e.,
\begin{equation*}
  \mathcal{O}_x|i\rangle|0\rangle\rightarrow|i\rangle|x_i\rangle.
\end{equation*}
The task is to prepare a quantum state that is $\epsilon$-close to
\begin{equation*}
  |x\rangle=\frac{1}{\sqrt{\|\boldsymbol x\|_1}}\sum_{i=0}^{n-1}\sqrt{x_i}|i\rangle
\end{equation*}
with a adjustable bound on the success probability $1-\delta^2$. The quantum algorithm for preparing this state contains two steps:

1. Prepare the initial state

(i) Start with a uniform superposition state and perform the black box to have
\begin{equation}
  \sum_{i=0}^{n-1}\frac{1}{\sqrt{n}}|i\rangle_1|x_i\rangle_a.
\end{equation}

(ii) Add a qubit and perform controlled rotation to yield

\begin{equation}
  \sum_{i=0}^{n-1}\frac{1}{\sqrt{n}}|i\rangle_1|x_i\rangle_a(\frac{\sqrt{x_i}}{|\boldsymbol x|_{\textrm{max}}}|0\rangle_2+\sqrt{1-\frac{|x_i|}{|\boldsymbol x|^2_{\textrm{max}}}}|1\rangle_2),
\end{equation}
where $|\boldsymbol x|_{\textrm{max}}=\max_{i}|\sqrt{x_i}|$ is a small constant.

(iii) Uncompute the black box to obtain
\begin{equation}
  \sum_{i=0}^{n-1}\frac{1}{\sqrt{n}}|i\rangle_1(\frac{\sqrt{x_i}}{|\boldsymbol x|_{\textrm{max}}}|0\rangle_2+\sqrt{1-\frac{|x_i|}{|\boldsymbol x|^2_{\textrm{max}}}}|1\rangle_2).
\end{equation}
Denote this state with $|\psi\rangle_{1,2}$, and this step can be regarded as a unitary operator $\boldsymbol U_a$ that $\boldsymbol U_a|0\rangle_{1,2}=|\psi\rangle_{1,2}$. Note that $|\psi\rangle_{1,2}$ can be rewritten  as follow,

\begin{equation}
|\psi\rangle_{1,2}=\sqrt{P_0}|\alpha\rangle_{1,2}+\sqrt{1-P_0}|\beta\rangle_{1,2}
\end{equation}
where
\begin{equation}
P_0=\frac{\|\boldsymbol x\|_1}{n|\boldsymbol x|_{\max}^{2}},
\end{equation}
\begin{equation}
  |\alpha\rangle_{1,2}=\sum_{i=0}^{n-1}\frac{\sqrt{x_i}}{\sqrt{\|\boldsymbol x\|_1}}|i\rangle_1|0\rangle_2,
\end{equation}
\begin{equation}\label{lesser interest}
  |\beta\rangle_{1,2}=\sum_{i=0}^{n-1}\frac{\sqrt{1-\frac{|x_i|}{|\boldsymbol x|^2_{\textrm{max}}}}}{\Upsilon}|i\rangle_1|1\rangle_2,
\end{equation}
\begin{equation}
\Upsilon=\sqrt{\sum_{i=0}^{n-1}\big|1-\frac{|x_i|}{|\boldsymbol x|^2_{\textrm{max}}}}\big|.
\end{equation}

2. Amplify the amplitude of getting $|\alpha\rangle$

In this step, we apply the fixed-point quantum search algorithm proposed in \cite{yoder2014fixed} to amplify the success probability with an adjustable bound. More specifically, define conditional phase shift operator

\begin{equation}
  \boldsymbol S_t^\varphi=\boldsymbol I_{1,2}+(e^{i\varphi}-1)\boldsymbol I_1\otimes|0\rangle\langle0|_2,
\end{equation}
and
\begin{equation}
  \boldsymbol S_a^\phi=\boldsymbol I_{1,2}+(e^{i\phi}-1)|0\rangle\langle0|_{1,2}.
\end{equation}
The algorithm performs the following sequence of generalized Grover operator
\begin{equation}
  \boldsymbol G(\phi_l,\varphi_l)\boldsymbol G(\phi_{l-1},\varphi_{l-1})\cdots \boldsymbol G(\phi_1,\varphi_1),
\end{equation}
where $\boldsymbol G(\phi_j,\varphi_j)=-\boldsymbol U_a\boldsymbol S_a^{\phi_j} \boldsymbol U_a^\dagger\boldsymbol S_t^{\varphi_j}$.
The condition on the phases $\{\varphi_j,\phi_j,1\leq j\leq l\}$ was indicated in \cite{yoder2014fixed}:
\begin{equation}
 \phi_j=\varphi_{l-j+1}=-2\textrm{arccot}(\sqrt{1-\gamma^2}\tan(2\pi j/L))
\end{equation}
where $L=2l+1,\gamma=T^{-1}_{1/L}(1/\delta),\delta\in(0,1)$ and $T_L(x)$ is the $L$-th Chebyshev polynomial of the first kind. After $l$ iterations, the final state, up to a
global phase, will be
\begin{equation}
  |\psi_l\rangle=\sqrt{P_L}|\alpha\rangle+\sqrt{1-P_L}|\beta\rangle
\end{equation}
where $P_L=1-\delta^2T_L^2[T_{1/L}(1/\delta)\sqrt{1-P_0}]$ is the success probability.  It was shown that for a given $\delta$ and a known lower bound $P_{\textrm{min}}$ of $P_0$, the following condition of $L$
\begin{equation}
  L\geq\frac{\textrm{log}(2/\delta)}{\sqrt{P_{\textrm{min}}}}
\end{equation}
can ensure $P_L\geq1-\delta^2$.

We now show that the $P_{\textrm{min}}$ can be provided by using amplitude estimation. More specifically, for any $\epsilon_0>0$, amplitude estimation can approximate the probability $P_0$ up to additive error $P_0\epsilon_0$ with $O(1/(\epsilon_0\sqrt{P_0}))$ uses of the standard Grover operator. Let the output of amplitude estimation to be $P'_0$, then  $P_0\geq\frac{P'_0}{1+\epsilon_0}$. Thus, we can take $\frac{P'_0}{1+\epsilon_0}$ as a low bound of $P_0$. Note that $\epsilon_0$ is the relative error of estimated $P_0$ and we only need the low bound of $P_0$, so we can set $\epsilon_0$ to be a small constant like $1/2$. Then the query complexity of this step is $O(\sqrt{n}/\sqrt{\|\boldsymbol x\|_1})$.

Reviewing the cost of amplitude estimation and fixed-point amplitude amplification, we now analyze the complexity of this approach. Clearly, the query complexity is $O(\frac{\textrm{log}(1/\delta)\sqrt{n}}{\sqrt{\|\boldsymbol x\|_1}})$.  The gate complexity of this approach is dominated by the gate complexity of fixed-point amplitude amplification which is given by the gate complexity of the generalized  Grover operator multiplied by the number of iterations.  The conditional phase shift operator $\boldsymbol S_a^{\phi_j}$ and $\boldsymbol S_t^{\varphi_j}$ can be implemented by using $O(\textrm{log}\,n)$ elementary gates. The gate complexity of $\boldsymbol U_a$ depend on the precision of controlled rotation. Note that performing the controlled rotation with error $\epsilon_r$ will generate a state that is $O(\sqrt{n}\epsilon_r)$-close to $|\psi\rangle$. Then, after performing the fixed-point amplitude amplification, we can get a state that is $O(\sqrt{n}\epsilon_r/\sqrt{P_0})$-close to $|\psi_l\rangle$, since the error is also amplified by the fixed-point amplitude amplification. To obtain overall error $O(\epsilon_p)$, $\epsilon_r$ should be $O(\epsilon_p\sqrt{P_0}/\sqrt{n})$. Also it is shown that the controlled rotation can be performed with error $\epsilon_r$ using $O(\textrm{polylog}(\frac{1}{\epsilon_r}))$ elementary gates. Thus, the gate complexity of $\boldsymbol U_a$ is $O\Big(\textrm{polylog}(\frac{\sqrt{n}}{\epsilon_p\sqrt{P_0}})\Big)$. Putting these all together, we can obtain a state, up to a global phase, that is $\epsilon_p$-close to $|x\rangle$ with success probability at least $1-\delta^2$, using $O\Big(\frac{\sqrt{n}\textrm{log}(1/\delta)}{\sqrt{\|\boldsymbol x\|_1}}\Big)$ queries of $\mathcal{O}_x$ and $O\Big(\frac{\sqrt{n}\textrm{log}(1/\delta)}{\sqrt{\|\boldsymbol x\|_1}}\textrm{polylog}(\frac{n}{\epsilon_p\sqrt{\|\boldsymbol x\|_1}})\Big)$ elementary gates.

\section{Proof of theorem \ref{the2}}

Now, we show how to implement the block-encoding of $\boldsymbol T_n$, i.e., the quantum state preparation operators and the controlled unitary, in the two date access models introduced in Sec.\ref{seciic}.

\subsection{Black-box model}\label{blackboxT}
In the black-box model, we are given the black box
\begin{equation*}
  \mathcal{O}_{\boldsymbol T_n}|i\rangle|k\rangle|0\rangle=|i\rangle|k\rangle|t_{i,k}\rangle,
  \qquad i,k=0,1,\cdots,n-1.
\end{equation*}
Then, we can construct black box $\mathcal{O}_1$ satisfying
\begin{equation*}
\begin{aligned}
  \mathcal{O}_1|j\rangle|0\rangle=|j\rangle|t_j+t_{-[(n-j)\,\textrm{mod}\,n]}\rangle, \qquad j=0,1,\cdots,n-1,
\end{aligned}
\end{equation*}
by calling the black box $\mathcal{O}_{\boldsymbol T_n}$ twice to query the elements in the site $|j,0\rangle$ and $|0,(n-j)\,\textrm{mod}\,n\rangle$ and following an addition operation \cite{cuccaro2004new,li2020efficient}. More specifically, $\mathcal{O}_1$ can be constructed as follows.
~\\
(1) Compute the index of site by using X-gates and quantum modular subtractor \cite{cuccaro2004new,li2020efficient}:
\begin{equation}
\begin{aligned}
   &|j\rangle_{a_1}|0\rangle_{a_2}|0\rangle_{b_1}|0\rangle_{b_2}|0\rangle_{a_3}|0\rangle_{b_3}\\
\rightarrow&|j\rangle_{a_1}|0\rangle_{a_2}|0\rangle_{b_1}|(n-j)\,\textrm{mod}\,n\rangle_{b_2}|0\rangle_{a_3}|0\rangle_{b_3}.
\end{aligned}
\end{equation}
(2) Perform $\mathcal{O}_{\boldsymbol T_n}$ on registers $\{a_1,a_2,a_3\}$ and $\{b_1,b_2,b_3\}$ respectively:
\begin{equation}
\begin{aligned}
  &|j\rangle_{a_1}|0\rangle_{a_2}|0\rangle_{b_1}|(n-j)\,\textrm{mod}\,n\rangle_{b_2}|0\rangle_{a_3}|0\rangle_{b_3}\\
\rightarrow&|j\rangle_{a_1}|0\rangle_{a_2}|0\rangle_{b_1}|(n-j)\,\textrm{mod}\,n\rangle_{b_2}|t_j\rangle_{a_3}|t_{-[(n-j)\,\textrm{mod}\,n]}\rangle_{b_3}.
\end{aligned}
\end{equation}
(3) Perform quantum addition operation on registers $\{a_3,b_3\}$ to yield
\begin{equation}
|j\rangle_{a_1}|0\rangle_{a_2}|0\rangle_{b_1}|(n-j)\,\textrm{mod}\,n\rangle_{b_2}|t_j\rangle_{a_3}|t_{j}+t_{-[(n-j)\,\textrm{mod}\,n]}\rangle_{b_3}.
\end{equation}
(4) Reverse the computation on registers $\{a_3,b_2\}$:
\begin{equation}
|j\rangle_{a_1}|0\rangle_{a_2}|0\rangle_{b_1}|0\rangle_{b_2}|0\rangle_{a_3}|t_{j}+t_{-[(n-j)\,\textrm{mod}\,n]}\rangle_{b_3}.
\end{equation}
The mapping on registers $a_1$ and $b_3$ is actually the $\mathcal{O}_1$.

Similarly, we can construct black box $\mathcal{O}_2$ satisfying
\begin{equation}
\mathcal{O}_2|j\rangle|0\rangle=|j\rangle|t_{j}-t_{-[(n-j)\,\textrm{mod}\,n]}\rangle, \qquad j=0,1,\cdots,n-1.
\end{equation}
Since the gates required for constructing $\mathcal{O}_1$ and $\mathcal{O}_2$ are negligible compared to other subroutines of the quantum algorithm, we did not consider their complexity in this paper. With these two black boxes $\mathcal{O}_1$ and $\mathcal{O}_2$, it is feasible to generate a controlled black box $\mathcal{O}_{1\wedge2}$ of the form
\begin{equation}
  |0\rangle\langle 0|\otimes\mathcal{O}_1+|1\rangle\langle 1|\otimes\mathcal{O}_2.
\end{equation}
Note that the black box $\mathcal{O}_{1\wedge2}$ query in superposition  acting as
\begin{equation}
\mathcal{O}_{1\wedge2} \frac{1}{\sqrt{2n}}(\sum_{j=0}^{n-1}|0\rangle|j\rangle|0\rangle+\sum_{j=0}^{n-1}|1\rangle|j\rangle|0\rangle)=\frac{1}{\sqrt{2n}}\sum_{j=0}^{2n-1}|j\rangle|\tilde{t}_j\rangle.
\end{equation}
Thus, using this black box, we can approximatively implement $\boldsymbol V_{(\nabla[\boldsymbol T_n])}$ by the steerable black-box quantum state preparation algorithm.  Similarly, we can approximatively implement $\boldsymbol V_{(\nabla[\boldsymbol T_n]^{\ast})}$ by constructing a black box $\mathcal{O}_{1\wedge2}^{*}$ that returns $\tilde{t}_{j}^{\ast}$.

When implementing $\textrm{select} \boldsymbol U_{T_n}$, directly using the controlled circuit may take $O(n\textrm{log}\,n)$ elementary gates. To make the implementation more efficient, we take an idea similar to that in Ref. \cite{zhou2017efficient}.

More specifically, note that
\begin{equation}
  \boldsymbol Z_{1}^{j}=\sum_{a=0}^{n-1}|(a+j)\,\textrm{mod}\,n\rangle\langle a|,
\end{equation}
\begin{equation}
  \boldsymbol Z_{-1}^{j}=(\sum_{b=0}^{n-1}|(b+j)\,\textrm{mod}\,n\rangle\langle b|)
      (\sum_{b=0}^{n-1-j}|b\rangle\langle b|-\sum_{b=n-j}^{n-1}|b\rangle\langle b|).
\end{equation}
Thus, the action of $\textrm{select} \boldsymbol U_{T_n}$ on the basis states is
\begin{widetext}
\begin{equation}
\textrm{select} \boldsymbol U_{T_n}|j\rangle|e\rangle=\left\{\begin{array}{ccc}
|j\rangle|(e+j)\,\textrm{mod}\,n\rangle &\quad 0\leq j\leq n-1,  &0\leq e\leq n-1 \\
|j\rangle|(e+j)\,\textrm{mod}\,n\rangle &\quad n\leq j\leq 2n-1, &0\leq e\leq 2n-1-j \\
-|j\rangle|(e+j)\,\textrm{mod}\,n\rangle&\quad n\leq j\leq 2n-1, &2n-j\leq e\leq n-1\\
\end{array}\right.
\end{equation}
\end{widetext}

Then, on the one hand, let
\begin{equation}
f_1(j,e)=\left\{\begin{array}{ccc}
0 &\quad 0\leq j\leq n-1,  &0\leq e\leq n-1 \\
0 &\quad n\leq j\leq 2n-1, &0\leq e\leq 2n-1-j \\
1 &\quad n\leq j\leq 2n-1, &2n-j\leq e\leq n-1.\\
\end{array}\right.
\end{equation}
To compute this classical function with a quantum circuit, we take the quantum comparator \cite{cuccaro2004new,li2020efficient} which comprises $O(\textrm{log}\,n)$ elementary gates. Suppose that $a$ and $b$ are two natural numbers, the quantum comparator output the result $c$ of the comparison of the two numbers, i.e., if $b\geq a, c=0$; otherwise, $c=1$. The specific quantum circuit of $\boldsymbol U_{f_1}$ is as follows.\\
(1) Prepare an initial quantum state
\begin{equation*}
\begin{aligned}
  &|j\rangle_{a_1}|e\rangle_{a_2}|0\rangle_{b_1}|0\rangle_{b_2}|0\rangle_{c_1}|0\rangle_{c_2}|-\rangle_{c_3}\\
  \rightarrow&
|j\rangle_{a_1}|e\rangle_{a_2}|n-1\rangle_{b_1}|2n-1-j\rangle_{b_2}|0\rangle_{c_1}|0\rangle_{c_2}|-\rangle_{c_3}.
\end{aligned}
\end{equation*}
(2) Perform quantum comparator on registers $\{a_1,b_1,c_1\}$ and $\{a_2,b_2,c_2\}$ respectively.\\
(3) Perform Toffoli gate on the registers $\{c_1,c_2,c_3\}$.\\
(4) Reverse the computation on registers $\{c_2,c_1,b_2,b_1\}$ in order.\\
The mapping on registers $a_1, a_2$ and $c_3$ is
\begin{equation}
  \boldsymbol U_{f_1}|j\rangle|e\rangle\frac{|0\rangle-|1\rangle}{\sqrt{2}}=(-1)^{f_1(j,e)}|j\rangle|e\rangle\frac{|0\rangle-|1\rangle}{\sqrt{2}}.
\end{equation}




On the other hand, using quantum modular adder \cite{cuccaro2004new,li2020efficient}, which requires $O(\textrm{log}\,n)$  elementary gates, we can implement
\begin{equation}
  \boldsymbol U_{add1}|j\rangle|e\rangle=|j\rangle|(e+j)\,\textrm{mod}\,n\rangle.
\end{equation}
Therefore, by using $\boldsymbol U_{add1}$ and $\boldsymbol U_{f_1}$, we can implement $\textrm{select} \boldsymbol U_{T_n}$ equivalently (Due to linearity, the implementation is available for any state). In summary, $\textrm{select} \boldsymbol U_{T_n}$ can be implemented in time $O(\textrm{poly}\textrm{log}\,n)$.


Now, we analyse the error incurred due to imperfect state preparation. Let
$\boldsymbol U_{\bar{V}}$ and $\boldsymbol U_{\bar{V}^\ast}$ be unitaries that perform the steerable black-box quantum state preparation with the black boxes $\mathcal{O}_{1\wedge2}$ and $\mathcal{O}_{1\wedge2}^{*}$. Define $\boldsymbol U_V$ and $\boldsymbol U_{V^\ast}$ as shown below,
\begin{equation}
  \boldsymbol U_V|0\rangle_{a_1}|0\rangle_{a_2}=\sqrt{1-\delta^2}|V_{(\nabla[ \boldsymbol T_n])}\rangle_{a_1}|0\rangle_{a_2}+\delta|\zeta\rangle_{a_1}|1\rangle_{a_2},
\end{equation}
\begin{equation}
  \boldsymbol U_{V^\ast}|0\rangle_{a_1}|0\rangle_{a_2}=\sqrt{1-\delta^2}|V_{(\nabla[ \boldsymbol T_n]^\ast)}\rangle_{a_1}|0\rangle_{a_2}+\delta|\zeta\rangle_{a_1}|1\rangle_{a_2},
\end{equation}
where $|\zeta\rangle=\sum_{j=0}^{2n-1}\zeta_j|j\rangle$ is a quantum state similar in form to Eq. (\ref{lesser interest}).
Obviously, $\boldsymbol U_{\bar{V}}$ and $\boldsymbol U_{\bar{V}^\ast}$ are $\epsilon_p$-approximation of $\boldsymbol U_V$ and $\boldsymbol U_{V^\ast}$ respectively. To simplify the notation, we rewrite Toeplitz matrices as $\boldsymbol T_n= \frac{1}{2}\sum_{j=0}^{2n-1}\tilde{t}_j \boldsymbol U_j$ by using the symbol $\boldsymbol U_j$ to represent $\boldsymbol Z_{1}^i$ and $\boldsymbol Z_{-1}^i$ in the order of Eq. (\ref{decT}). Similarly, the operator defined by Eq. (\ref{controluT}) can be rewritten as $\textrm{select} \boldsymbol U=\sum_{j=0}^{2n-1}|j\rangle\langle j|_{a_1}\otimes \boldsymbol U_j$, where $\boldsymbol U_j$ perform on the register $s$. Let $|0\rangle_{a}\equiv|0\rangle_{a_1}|0\rangle_{a_2}$, we have

\begin{equation}\label{noteq}
\begin{aligned}
&\|\boldsymbol M- \frac{\chi}{2} (\langle0|_a\boldsymbol U_{\bar{V}^\ast}^\dagger\otimes\boldsymbol I_{s})(\textrm{select} \boldsymbol U\otimes \boldsymbol I_{a_2})(\boldsymbol U_{\bar{V}}|0\rangle_a\otimes \boldsymbol I_{s})\|\\
&\leq\|\boldsymbol M- \frac{\chi}{2} (\langle0|_a\boldsymbol U_{V^\ast}^\dagger\otimes\boldsymbol I_{s})(\textrm{select} \boldsymbol U\otimes \boldsymbol I_{a_2})(\boldsymbol U_{V}|0\rangle_a\otimes \boldsymbol I_{s})\|\\
&+\|\frac{\chi}{2} (\langle0|_a\boldsymbol U_{V^\ast}^\dagger\otimes\boldsymbol I_{s})(\textrm{select} \boldsymbol U\otimes \boldsymbol I_{a_2})(\boldsymbol U_{V}|0\rangle_a\otimes \boldsymbol I_{s})\\
&-\frac{\chi}{2}(\langle0|_a\boldsymbol U_{\bar{V}^\ast}^\dagger\otimes\boldsymbol I_{s})(\textrm{select} \boldsymbol U\otimes \boldsymbol I_{a_2})(\boldsymbol U_{\bar{V}}|0\rangle_a\otimes \boldsymbol I_{s})\|
\end{aligned}
\end{equation}

Note that,
\begin{widetext}
\begin{equation}
\begin{aligned}
&(\langle0|_{a}\boldsymbol U_{V^\ast}^\dagger\otimes \boldsymbol I_s)(\textrm{select} \boldsymbol U\otimes\boldsymbol I_{a_2})(\boldsymbol U_{V}|0\rangle_{a}\otimes \boldsymbol I_s)\\
&=(\langle0|_{a_{1}}\langle0|_{a_{2}}\boldsymbol U_{V^\ast}^\dagger\otimes \boldsymbol I_s)(\sqrt{1-\delta^2}\sum_{j=0}^{2n-1}\sqrt{\frac{\tilde{t}_j}{\chi}}|j\rangle_{a_1}|0\rangle_{a_2}\otimes\boldsymbol U_j+\delta\sum_{j=0}^{2n-1}\zeta_j|j\rangle_{a_1}|1\rangle_{a_2}\otimes\boldsymbol U_j)\\
&=(\sqrt{1-\delta^2}\!\!\sum_{j=0}^{2n-1}\!\!\sqrt{\frac{\tilde{t}_j}{\chi}}\langle j|_{a_{1}}\langle0|_{a_{2}}\!\!+\delta\!\!\sum_{j=0}^{2n-1}\!\!\zeta_j\langle j|_{a_{1}}\langle1|_{a_{2}})(\sqrt{1-\delta^2}\!\!\sum_{j=0}^{2n-1}\!\!\sqrt{\frac{\tilde{t}_j}{\chi}}|j\rangle_{a_{1}}|0\rangle_{a_{2}}\otimes\!\boldsymbol U_j+\delta\!\!\sum_{j=0}^{2n-1}\!\!\zeta_j|j\rangle_{a_{1}}|1\rangle_{a_{2}}\otimes\!\boldsymbol U_j)\\
&=(1-\delta^2)\sum_{j=0}^{2n-1}\frac{\tilde{t}_j}{\chi}\boldsymbol U_j+\delta^2\sum_{j=0}^{2n-1}\zeta_j^2\boldsymbol U_j
\end{aligned}
\end{equation}
\end{widetext}
Then, computing the first term on the right hand side of the inequality,
\begin{equation}
\begin{aligned}
  &\|\boldsymbol M- \frac{\chi}{2} (\langle0|_a\boldsymbol U_{V^\ast}^\dagger\otimes\boldsymbol I_{s})(\textrm{select} \boldsymbol U\otimes \boldsymbol I_{a_2})(\boldsymbol U_{V}|0\rangle_a\otimes \boldsymbol I_{s})\|\\
  &=\|\frac{1}{2}\sum_{j=0}^{2n-1}\tilde{t}_j \boldsymbol U_j-\frac{1}{2}(1-\delta^2)\sum_{j=0}^{2n-1}\tilde{t}_j\boldsymbol U_j-\frac{1}{2}\chi\delta^2\sum_{j=0}^{2n-1}\zeta_j^2\boldsymbol U_j\|\\
  &=\|\frac{1}{2}\delta^2\sum_{j=0}^{2n-1}\tilde{t}_j\boldsymbol U_j-\frac{1}{2}\chi\delta^2\sum_{j=0}^{2n-1}\zeta_j^2\boldsymbol U_j\|\\
  &\leq\chi\delta^2
\end{aligned}
\end{equation}

In addition, since $\boldsymbol U_{\bar{V}}$ and $\boldsymbol U_{\bar{V}^\ast}$ are $\epsilon_p$-approximation of $\boldsymbol U_V$ and $\boldsymbol U_{V^\ast}$ respectively, the second term on the right hand side of the Eq.(\ref{noteq}) can be bounded by $\chi\epsilon_p$. Let $\chi\delta^2$ and $\chi\epsilon_p$ not larger than $\epsilon/2$, we can get the result (i) of theorem 2.

\subsection{QRAM data structure model}

For QRAM data structure model, the quantum state preparation operators can be implemented using the method of \cite{kerenidis2016quantum}. More specifically,

\begin{lemma}\label{datastructureT} (\cite{kerenidis2016quantum}) Suppose that  $\boldsymbol x\in \mathbb{C}^{n\times 1}$ is stored in a QRAM data structure, i.e., the entry $x_i$ is stored in $i$-th leaf of a binary tree, the internal node of the tree stores the sum of the modulus of elements in the subtree rooted at it.  Then, there is a quantum algorithm that can generate an $\epsilon_p$-approximation of $|x\rangle=\frac{1}{\sqrt{\|\boldsymbol x\|_1}}\sum_{i=0}^{n-1}\sqrt{x_i}|i\rangle$ with gate complexity $O(\textrm{polylog}(n/\epsilon_p))$.
\end{lemma}

Obviously, if $\{\tilde{t}_j\}_{j=0}^{n-1}$  and $\{\tilde{t}_j\}_{j=n}^{2n-1}$ are stored in such data structure respectively, there are two unitaries that generate the states:
\begin{equation}
  \boldsymbol U_1|0\rangle=\frac{1}{\sqrt{\chi_1}}\sum_{j=0}^{n-1}\sqrt{\tilde{t}_j}|j\rangle, \quad  \boldsymbol U_2|0\rangle=\frac{1}{\sqrt{\chi_2}}\sum_{j=0}^{n-1}\sqrt{\tilde{t}_{j+n}}|j\rangle,
\end{equation}
where $\chi_1=\sum_{j=0}^{n-1}|\tilde{t}_j|$, $\chi_2=\sum_{j=n}^{2n-1}|\tilde{t}_j|$.

Since $\chi_1$ and $\chi_2$ are known, which are stored in the root of the  binary trees, we can prepare  a state
\begin{equation*}
  \frac{\sqrt{\chi_1}}{\sqrt{\chi_{\boldsymbol T_{n}}}}|0\rangle|0\rangle+\frac{\sqrt{\chi_2}}{\sqrt{\chi_{\boldsymbol T_{n}}}}|1\rangle|0\rangle.
\end{equation*}
Then, performing a controlled unitary $|0\rangle\langle0|\otimes \boldsymbol U_1+|1\rangle\langle1|\otimes \boldsymbol U_2$, we can get the state $\frac{1}{\sqrt{\chi_{\boldsymbol T_{n}}}}\sum_{j=0}^{2n-1}\sqrt{\tilde{t}_j}|j\rangle $.
Thus, $\boldsymbol V_{(\nabla[\boldsymbol T_n])}$ can be implemented in QARM data structure model with complexity $O(\textrm{polylog}(n/\epsilon_p))$. Similarly, we can implement $\boldsymbol V_{(\nabla[\boldsymbol T_n]^{\ast})}$ with the same cost.

In addition, the $\textrm{select} \boldsymbol U_{T_n}$ can be implemented in the same way as in the black-box model. Taking into account the amplification of the error, we can implement the block-encoding with complexity $O(\textrm{polylog}(n\chi_{\boldsymbol T_n}/\epsilon))$ in QRAM data structure model. Besides, according to the constructed data structure, the memory cost in this data access model is $O(n)$.

\section{Proof of  corollary \ref{the3}}

\subsection{Black-box model}\label{blackboxTL}
For a Toeplitz-like matrix $\boldsymbol T_L$, given a black box  $O_{\boldsymbol T_L}$ that query the $k$-th non-zero element of the $i$-th row of $\boldsymbol T_L$,
\begin{equation*}
\mathcal{O}_{\boldsymbol T_L}|i,k\rangle|0\rangle=|i,k\rangle|\tau_{i,k}\rangle,
\end{equation*}
the following map can be performed by querying the black box $O_{\boldsymbol T_L}$ twice,
\begin{equation*}
\mathcal{O}_{ \tilde{\boldsymbol T}_L}|i,k\rangle|0\rangle=|i,k\rangle|\tau_{(i-1)\textrm{mod}n,k}-(-1)^{g(k)}\tau_{i,(k+1)\textrm{mod}n}\rangle.
\end{equation*}
$\mathcal{O}_{ \tilde{\boldsymbol T}_L}$ actually returns the $k$-th non-zero element of the $i$-th row of the the Sylvester displacements of Toeplitz-like matrices, i.e., $\tilde{\tau}_{i,k}$.

In addition, if a black box that computes the positions of the distinct elements (the element is different from the previous element on the same diagonal) on diagonals of the Toeplitz-like matrices is provided,  we can construct a black box that computes the positions of nonzero elements of the the Sylvester displacements of Toeplitz-like matrices, i.e.,
\begin{equation*}
  \mathcal{O}^p_{\tilde{\boldsymbol T}_L}|i,k\rangle=|i,f(i,k)\rangle.
\end{equation*}
where the function $f(i,k)$  gives the column index of the $k$-th nonzero element in row $i$ of $\nabla_{\boldsymbol Z_1, \boldsymbol Z_{-1}}[\boldsymbol T_L]$.

When implementing the state preparation operators, if the steerable black-box quantum state preparation is directly called, the query complexity should be $O(n)$. To overcome this obstacle, we first prepare a uniform superposition state that only represents the position of the non-zero elements of $\nabla_{\boldsymbol Z_1, \boldsymbol Z_{-1}}[\boldsymbol T_L]$. The specific process is as follows.

Prepare an initial state as:
\begin{equation}
 \frac{1}{\sqrt{2}}(|0\rangle_1+|1\rangle_1)|0\rangle_2|0\rangle_3.
\end{equation}

Apply Hadamard gates to Reg 2 and Reg 3 controlled by Reg 1:
\begin{equation}
\begin{aligned}
  &\frac{1}{\sqrt{2}}(|0\rangle_1+|1\rangle_1)|0\rangle_2|0\rangle_3\\
  &\xrightarrow{|0\rangle\langle0|_1\otimes \boldsymbol I_2\otimes\boldsymbol H_3^{\otimes \textrm{log}n}+|1\rangle\langle1|_1\otimes \boldsymbol H_2^{\otimes \textrm{log}n}\otimes\boldsymbol I_3^{\otimes \textrm{log}(n/d)}\otimes\boldsymbol H_3^{\otimes \textrm{log}d}}\\
  &\frac{1}{\sqrt{2n}}|0\rangle_1|0\rangle_2\sum_{k=0}^{n-1}|k\rangle_3+\frac{1}{\sqrt{2nd}}|1\rangle_1\sum_{i=0}^{n-1}|i\rangle_2\sum_{k=0}^{d-1}|k\rangle_3.
\end{aligned}
\end{equation}

Using a controlled-$\mathcal{O}^p_{ \tilde{\boldsymbol T}_L}$, i.e., $|0\rangle\langle0|_1\otimes \boldsymbol I_{2,3}+|1\rangle\langle1|_1\otimes \mathcal{O}^p_{ \tilde{\boldsymbol T}_L}$, we can prepare:
\begin{equation}
\frac{1}{\sqrt{2n}}|0\rangle_1|0\rangle_2\sum_{k=0}^{n-1}|k\rangle_3+\frac{1}{\sqrt{2nd}}|1\rangle_1\sum_{i=0}^{n-1}|i\rangle_2\!\!\!\sum_{\{f(i,k)|\tilde{m}_{i,f(i,k)}\neq0\}}\!\!\!\!\!\!\!\!\!\!\!\!\!\!\!|f(i,k)\rangle_3.
\end{equation}

Add an ancillary qubit and perform a controlled rotation:

\begin{equation}
\begin{aligned}
&\frac{1}{\sqrt{2n}}|0\rangle_1|0\rangle_2\sum_{k=0}^{n-1}|k\rangle_3(\frac{1}{\sqrt{d}}|0\rangle_4+\sqrt{1-\frac{1}{d}}|1\rangle_4)\\
&+\frac{1}{\sqrt{2nd}}|1\rangle_1\sum_{i=0}^{n-1}|i\rangle_2\!\!\!\sum_{\{f(i,k)|\tilde{m}_{i,f(i,k)}\neq0\}}\!\!\!\!\!\!\!\!\!\!\!\!\!\!\!|f(i,k)\rangle_3|0\rangle_4.
\end{aligned}
\end{equation}

Amplify the amplitude of $|0\rangle_4$. Since the amplitude is known,  it can be amplified to exactly 1 by using Long's amplitude amplification with zero theoretical failure rate \cite{long2001grover}. The obtained state is denoted as $|\Phi_{\textrm{intm}}\rangle$:
\begin{equation}
\begin{aligned}
&|\Phi_{\textrm{intm}}\rangle=\frac{1}{\sqrt{n(d+1)}}|0\rangle_1|0\rangle_2\sum_{k=0}^{n-1}|k\rangle_3|0\rangle_4\\
&+\frac{1}{\sqrt{n(d+1)}}|1\rangle_1\sum_{i=0}^{n-1}|i\rangle_2\!\!\!\sum_{\{f(i,k)|\tilde{m}_{i,f(i,k)}\neq0\}}\!\!\!\!\!\!\!\!\!\!\!\!\!\!\!\!|f(i,k)\rangle_3|0\rangle_4.
\end{aligned}
\end{equation}

We run Long's amplitude amplification again to get the quantum state $|\Phi_{\textrm{init}}\rangle$:

\begin{equation}
\begin{aligned}
&|\Phi_{\textrm{init}}\rangle=\frac{1}{\sqrt{n+(n-1)d}}|0\rangle_1|0\rangle_2\sum_{k=0}^{n-1}|k\rangle_3|0\rangle_4\\
&+\frac{1}{\sqrt{n+(n-1)d}}|1\rangle_1\sum_{i=1}^{n-1}|i\rangle_2\!\!\!\sum_{\{f(i,k)|\tilde{m}_{i,f(i,k)}\neq0\}}\!\!\!\!\!\!\!\!\!\!\!\!\!\!\!|f(i,k)\rangle_3|0\rangle_4.
\end{aligned}
\end{equation}

Note that the success probability of getting $|\Phi_{\textrm{intm}}\rangle$ is $\frac{d+1}{2d}\geq\frac{1}{2}$, and the success probability of getting $|\Phi_{\textrm{init}}\rangle$ is $\frac{(n-1)d+n}{n(d+1)}\geq\frac{1}{2}$, thus only a few iterations are required for the amplitude amplifications.

With the quantum state $|\Phi_{\textrm{init}}\rangle$ and the black box $O_{\tilde{\boldsymbol T}_L}$, we can approximatively implement $\boldsymbol V_{(\nabla[\boldsymbol T_L])}$ and $\boldsymbol V_{(\nabla[\boldsymbol T_L]^{\ast})}$by the steerable black-box quantum state preparation algorithm. The query complexity is  $O\Big(\frac{\sqrt{nd}\textrm{log}(1/\delta)}{\sqrt{\chi_{\boldsymbol T_L}}}\Big)$, and $O\Big(\frac{\sqrt{nd}\textrm{log}(1/\delta)}{\sqrt{\chi_{\boldsymbol T_L}}}\textrm{polylog}(\frac{nd}{\sqrt{\chi_{\boldsymbol T_L}}\epsilon_p})\Big)$ elementary gates are required.

To implement $\textrm{select} \boldsymbol U_{T_L}$, we first observe its action on the basis states. Notice that
\begin{equation}
 \textrm{select} \boldsymbol U_{T_L}|i\rangle|k\rangle|e\rangle=\left\{\begin{array}{cc}
|i\rangle|k\rangle|(i+e-k-1)\,\textrm{mod}\,n\rangle\qquad & \\ \qquad\qquad\textrm{where}\;0\leq e\leq k, \\
-|i\rangle|k\rangle|(i+e-k-1)\,\textrm{mod}\,n\rangle\qquad & \\ \qquad\qquad \textrm{where}\; k< e\leq n-1.\\
\end{array}\right.
\end{equation}

Then, on the one hand, let
\begin{equation}
 f_2(k,e)=\left\{\begin{array}{cc}
0  &\quad\quad 0\leq e\leq k, \\
1 & \quad k< e\leq n-1.\\
\end{array}\right.
\end{equation}
Similar to the calculation of $f_1$, we can construct a quantum circuit to implement
\begin{equation}
  \boldsymbol U_{f_2}|k\rangle|e\rangle\frac{|0\rangle-|1\rangle}{\sqrt{2}}=(-1)^{f_2(k,e)}|k\rangle|e\rangle\frac{|0\rangle-|1\rangle}{\sqrt{2}}.
\end{equation}

On the other hand, using quantum adders \cite{cuccaro2004new,li2020efficient}, which requires $O(\textrm{log}\,n)$  elementary gates, we can implement
\begin{equation}
  \boldsymbol U_{add2}|i\rangle|k\rangle|e\rangle=|i\rangle|k\rangle|(i+e-k-1)\,\textrm{mod}\,n\rangle.
\end{equation}
Therefore, $\textrm{select} \boldsymbol U_{T_L}$ can be implemented by $\boldsymbol U_{add2}$ and $\boldsymbol I\otimes\boldsymbol U_{f_2}$ in time $O(\textrm{poly}\textrm{log}\,n)$.

Based on the above conclusions and following the error analysis in \ref{blackboxT}, we can infer the result (i) of  corollary \ref{the3}.

\subsection{QRAM data structure model}
For QRAM data structure model, the quantum state preparation operators can be implemented as follows,

\begin{lemma}\label{datastructureTL} Let $\boldsymbol T_L\in \mathbb{C}^{n\times n}$, $\|\tilde{\tau}_{i,\cdot}\|_1$ be the 1-norm of the $i$-th row of $\nabla_{\boldsymbol Z_1, \boldsymbol Z_{-1}}[\boldsymbol T_L]$. Suppose that $\nabla_{\boldsymbol Z_1, \boldsymbol Z_{-1}}[\boldsymbol T_L]$ is stored in a QRAM data structure, more specifically, for the $i$-th row of $\nabla_{\boldsymbol Z_1, \boldsymbol Z_{-1}}[\boldsymbol T_L]$, the entry $\tilde{\tau}_{i,k}$ is stored in $k$-th leaf of a binary tree, the internal node of the tree stores the sum of the modulus of elements in the subtree rooted at it, and an additional binary tree of which the $i$-th leaf stores $\|\tilde{\tau}_{i,\cdot}\|_1$.  Then, there is a quantum algorithm that can perform the following maps with $\epsilon_p$-precision in time $O(\textrm{polylog}(n/\epsilon_p))$:
\begin{equation}
  \boldsymbol P:|i\rangle|0\rangle\mapsto\frac{\sum_{k=0}^{n-1}\sqrt{\tilde{\tau}_{i,k}}|i\rangle|k\rangle}{\sqrt{\|\tilde{\tau}_{i,\cdot}\|_1}},
\end{equation}
\begin{equation}
  \boldsymbol P':|i\rangle|0\rangle\mapsto\frac{\sum_{k=0}^{n-1}\sqrt{\tilde{\tau}_{i,k}^{\ast}}|i\rangle|k\rangle}{\sqrt{\|\tilde{\tau}_{i,\cdot}\|_1}},
\end{equation}
\begin{equation}
  \boldsymbol Q:|0\rangle|k\rangle\mapsto\frac{\sum_{i=0}^{n-1}\sqrt{\|\tilde{\tau}_{i,\cdot}\|_1}|i\rangle|k\rangle}{\sqrt{\chi_{\boldsymbol T_L}}}.
\end{equation}
\end{lemma}

This conclusion can be directly derived from the results in \cite{kerenidis2016quantum}. Obviously,

\begin{equation}
\boldsymbol V_{(\nabla[\boldsymbol T_L])}|0\rangle|0\rangle=\boldsymbol P\boldsymbol Q|0\rangle|0\rangle=\frac{1}{\sqrt{\chi_{\boldsymbol T_L}}}\sum_{i=0}^{n-1}\sum_{k=0}^{n-1}\sqrt{\tilde{\tau}_{i,k}}|i\rangle|k\rangle.
\end{equation}
Similarly, we can efficiently implement $\boldsymbol V_{(\nabla[\boldsymbol T_L]^{\ast})}$ with $\boldsymbol P'$ and $\boldsymbol Q$.

Note that the $\textrm{select} \boldsymbol U_{T_L}$ can be implemented in the same way as in \ref{blackboxTL}. Taking into account the amplification of the error, we can implement the block-encoding with complexity $O(\textrm{polylog}(n\chi_{\boldsymbol T_L}/\epsilon))$ in QRAM data structure model. Moreover, the memory cost in this data access model is easy to calculate as $O(dn\textrm{log}\,n)$.

\bibliography{block_based_re}

\begin{thebibliography}{59}%
\makeatletter
\providecommand \@ifxundefined [1]{%
 \@ifx{#1\undefined}
}%
\providecommand \@ifnum [1]{%
 \ifnum #1\expandafter \@firstoftwo
 \else \expandafter \@secondoftwo
 \fi
}%
\providecommand \@ifx [1]{%
 \ifx #1\expandafter \@firstoftwo
 \else \expandafter \@secondoftwo
 \fi
}%
\providecommand \natexlab [1]{#1}%
\providecommand \enquote  [1]{``#1''}%
\providecommand \bibnamefont  [1]{#1}%
\providecommand \bibfnamefont [1]{#1}%
\providecommand \citenamefont [1]{#1}%
\providecommand \href@noop [0]{\@secondoftwo}%
\providecommand \href [0]{\begingroup \@sanitize@url \@href}%
\providecommand \@href[1]{\@@startlink{#1}\@@href}%
\providecommand \@@href[1]{\endgroup#1\@@endlink}%
\providecommand \@sanitize@url [0]{\catcode `\\12\catcode `\$12\catcode
  `\&12\catcode `\#12\catcode `\^12\catcode `\_12\catcode `\%12\relax}%
\providecommand \@@startlink[1]{}%
\providecommand \@@endlink[0]{}%
\providecommand \url  [0]{\begingroup\@sanitize@url \@url }%
\providecommand \@url [1]{\endgroup\@href {#1}{\urlprefix }}%
\providecommand \urlprefix  [0]{URL }%
\providecommand \Eprint [0]{\href }%
\providecommand \doibase [0]{http://dx.doi.org/}%
\providecommand \selectlanguage [0]{\@gobble}%
\providecommand \bibinfo  [0]{\@secondoftwo}%
\providecommand \bibfield  [0]{\@secondoftwo}%
\providecommand \translation [1]{[#1]}%
\providecommand \BibitemOpen [0]{}%
\providecommand \bibitemStop [0]{}%
\providecommand \bibitemNoStop [0]{.\EOS\space}%
\providecommand \EOS [0]{\spacefactor3000\relax}%
\providecommand \BibitemShut  [1]{\csname bibitem#1\endcsname}%
\let\auto@bib@innerbib\@empty
\bibitem [{\citenamefont {Bennett}\ and\ \citenamefont
  {Brassard}(2014)}]{bennett2014quantum}%
  \BibitemOpen
  \bibfield  {author} {\bibinfo {author} {\bibfnamefont {C.~H.}\ \bibnamefont
  {Bennett}}\ and\ \bibinfo {author} {\bibfnamefont {G.}~\bibnamefont
  {Brassard}},\ }\href@noop {} {\bibfield  {journal} {\bibinfo  {journal}
  {Theor. Comput. Sci.}\ }\textbf {\bibinfo {volume} {560}},\ \bibinfo {pages}
  {7} (\bibinfo {year} {2014})}\BibitemShut {NoStop}%
\bibitem [{\citenamefont {Wei}\ \emph {et~al.}(2020)\citenamefont {Wei},
  \citenamefont {Cai}, \citenamefont {Wang}, \citenamefont {Qin}, \citenamefont
  {Gao},\ and\ \citenamefont {Wen}}]{wei2020error}%
  \BibitemOpen
  \bibfield  {author} {\bibinfo {author} {\bibfnamefont {C.-Y.}\ \bibnamefont
  {Wei}}, \bibinfo {author} {\bibfnamefont {X.-Q.}\ \bibnamefont {Cai}},
  \bibinfo {author} {\bibfnamefont {T.-Y.}\ \bibnamefont {Wang}}, \bibinfo
  {author} {\bibfnamefont {S.-J.}\ \bibnamefont {Qin}}, \bibinfo {author}
  {\bibfnamefont {F.}~\bibnamefont {Gao}}, \ and\ \bibinfo {author}
  {\bibfnamefont {Q.-Y.}\ \bibnamefont {Wen}},\ }\href@noop {} {\bibfield
  {journal} {\bibinfo  {journal} {IEEE Journal on Selected Areas in
  Communications}\ }\textbf {\bibinfo {volume} {38}},\ \bibinfo {pages} {517}
  (\bibinfo {year} {2020})}\BibitemShut {NoStop}%
\bibitem [{\citenamefont {Xu}\ and\ \citenamefont {Jiang}(2021)}]{xu2021novel}%
  \BibitemOpen
  \bibfield  {author} {\bibinfo {author} {\bibfnamefont {G.-B.}\ \bibnamefont
  {Xu}}\ and\ \bibinfo {author} {\bibfnamefont {D.-H.}\ \bibnamefont {Jiang}},\
  }\href@noop {} {\bibfield  {journal} {\bibinfo  {journal} {Quantum
  Information Processing}\ }\textbf {\bibinfo {volume} {20}},\ \bibinfo {pages}
  {128} (\bibinfo {year} {2021})}\BibitemShut {NoStop}%
\bibitem [{\citenamefont {Gao}\ \emph {et~al.}(2019)\citenamefont {Gao},
  \citenamefont {Qin}, \citenamefont {Huang},\ and\ \citenamefont
  {Wen}}]{gao2019quantum}%
  \BibitemOpen
  \bibfield  {author} {\bibinfo {author} {\bibfnamefont {F.}~\bibnamefont
  {Gao}}, \bibinfo {author} {\bibfnamefont {S.}~\bibnamefont {Qin}}, \bibinfo
  {author} {\bibfnamefont {W.}~\bibnamefont {Huang}}, \ and\ \bibinfo {author}
  {\bibfnamefont {Q.}~\bibnamefont {Wen}},\ }\href@noop {} {\bibfield
  {journal} {\bibinfo  {journal} {SCIENCE CHINA Physics, Mechanics \&
  Astronomy}\ }\textbf {\bibinfo {volume} {62}},\ \bibinfo {pages} {70301}
  (\bibinfo {year} {2019})}\BibitemShut {NoStop}%
\bibitem [{\citenamefont {Fitzsimons}(2017)}]{fitzsimons2017private}%
  \BibitemOpen
  \bibfield  {author} {\bibinfo {author} {\bibfnamefont {J.~F.}\ \bibnamefont
  {Fitzsimons}},\ }\href@noop {} {\bibfield  {journal} {\bibinfo  {journal}
  {npj Quantum Information}\ }\textbf {\bibinfo {volume} {3}},\ \bibinfo
  {pages} {1} (\bibinfo {year} {2017})}\BibitemShut {NoStop}%
\bibitem [{\citenamefont {He}\ \emph {et~al.}(2018)\citenamefont {He},
  \citenamefont {Sun}, \citenamefont {Yang},\ and\ \citenamefont
  {Yuan}}]{he2018exact}%
  \BibitemOpen
  \bibfield  {author} {\bibinfo {author} {\bibfnamefont {X.}~\bibnamefont
  {He}}, \bibinfo {author} {\bibfnamefont {X.}~\bibnamefont {Sun}}, \bibinfo
  {author} {\bibfnamefont {G.}~\bibnamefont {Yang}}, \ and\ \bibinfo {author}
  {\bibfnamefont {P.}~\bibnamefont {Yuan}},\ }\href@noop {} {\bibfield
  {journal} {\bibinfo  {journal} {arXiv preprint arXiv:1801.05717}\ } (\bibinfo
  {year} {2018})}\BibitemShut {NoStop}%
\bibitem [{\citenamefont {Chen}\ \emph {et~al.}(2020)\citenamefont {Chen},
  \citenamefont {Ye},\ and\ \citenamefont {Li}}]{chen2020characterization}%
  \BibitemOpen
  \bibfield  {author} {\bibinfo {author} {\bibfnamefont {W.}~\bibnamefont
  {Chen}}, \bibinfo {author} {\bibfnamefont {Z.}~\bibnamefont {Ye}}, \ and\
  \bibinfo {author} {\bibfnamefont {L.}~\bibnamefont {Li}},\ }\href@noop {}
  {\bibfield  {journal} {\bibinfo  {journal} {Physical Review A}\ }\textbf
  {\bibinfo {volume} {101}},\ \bibinfo {pages} {022325} (\bibinfo {year}
  {2020})}\BibitemShut {NoStop}%
\bibitem [{\citenamefont {Harrow}\ \emph {et~al.}(2009)\citenamefont {Harrow},
  \citenamefont {Hassidim},\ and\ \citenamefont {Lloyd}}]{eq1}%
  \BibitemOpen
  \bibfield  {author} {\bibinfo {author} {\bibfnamefont {A.~W.}\ \bibnamefont
  {Harrow}}, \bibinfo {author} {\bibfnamefont {A.}~\bibnamefont {Hassidim}}, \
  and\ \bibinfo {author} {\bibfnamefont {S.}~\bibnamefont {Lloyd}},\
  }\href@noop {} {\bibfield  {journal} {\bibinfo  {journal} {Physical review
  letters}\ }\textbf {\bibinfo {volume} {103}},\ \bibinfo {pages} {150502}
  (\bibinfo {year} {2009})}\BibitemShut {NoStop}%
\bibitem [{\citenamefont {Zhao}\ \emph {et~al.}(2021)\citenamefont {Zhao},
  \citenamefont {Zhao}, \citenamefont {Rebentrost},\ and\ \citenamefont
  {Fitzsimons}}]{zhao2021compiling}%
  \BibitemOpen
  \bibfield  {author} {\bibinfo {author} {\bibfnamefont {L.}~\bibnamefont
  {Zhao}}, \bibinfo {author} {\bibfnamefont {Z.}~\bibnamefont {Zhao}}, \bibinfo
  {author} {\bibfnamefont {P.}~\bibnamefont {Rebentrost}}, \ and\ \bibinfo
  {author} {\bibfnamefont {J.}~\bibnamefont {Fitzsimons}},\ }\href@noop {}
  {\bibfield  {journal} {\bibinfo  {journal} {Quantum Machine Intelligence}\
  }\textbf {\bibinfo {volume} {3}},\ \bibinfo {pages} {1} (\bibinfo {year}
  {2021})}\BibitemShut {NoStop}%
\bibitem [{\citenamefont {Gily{\'e}n}\ \emph {et~al.}(2019)\citenamefont
  {Gily{\'e}n}, \citenamefont {Su}, \citenamefont {Low},\ and\ \citenamefont
  {Wiebe}}]{gilyen2019quantum}%
  \BibitemOpen
  \bibfield  {author} {\bibinfo {author} {\bibfnamefont {A.}~\bibnamefont
  {Gily{\'e}n}}, \bibinfo {author} {\bibfnamefont {Y.}~\bibnamefont {Su}},
  \bibinfo {author} {\bibfnamefont {G.~H.}\ \bibnamefont {Low}}, \ and\
  \bibinfo {author} {\bibfnamefont {N.}~\bibnamefont {Wiebe}},\ }in\ \href@noop
  {} {\emph {\bibinfo {booktitle} {Proceedings of the 51st Annual ACM SIGACT
  Symposium on Theory of Computing}}}\ (\bibinfo {organization} {ACM},\
  \bibinfo {year} {2019})\ pp.\ \bibinfo {pages} {193--204}\BibitemShut
  {NoStop}%
\bibitem [{\citenamefont {Yu}\ \emph {et~al.}(2016)\citenamefont {Yu},
  \citenamefont {Gao}, \citenamefont {Wang},\ and\ \citenamefont
  {Wen}}]{yu2016quantum}%
  \BibitemOpen
  \bibfield  {author} {\bibinfo {author} {\bibfnamefont {C.-H.}\ \bibnamefont
  {Yu}}, \bibinfo {author} {\bibfnamefont {F.}~\bibnamefont {Gao}}, \bibinfo
  {author} {\bibfnamefont {Q.-L.}\ \bibnamefont {Wang}}, \ and\ \bibinfo
  {author} {\bibfnamefont {Q.-Y.}\ \bibnamefont {Wen}},\ }\href@noop {}
  {\bibfield  {journal} {\bibinfo  {journal} {Physical Review A}\ }\textbf
  {\bibinfo {volume} {94}},\ \bibinfo {pages} {042311} (\bibinfo {year}
  {2016})}\BibitemShut {NoStop}%
\bibitem [{\citenamefont {Yu}\ \emph {et~al.}(2019{\natexlab{a}})\citenamefont
  {Yu}, \citenamefont {Gao}, \citenamefont {Lin},\ and\ \citenamefont
  {Wang}}]{yu2019quantumprincipal}%
  \BibitemOpen
  \bibfield  {author} {\bibinfo {author} {\bibfnamefont {C.-H.}\ \bibnamefont
  {Yu}}, \bibinfo {author} {\bibfnamefont {F.}~\bibnamefont {Gao}}, \bibinfo
  {author} {\bibfnamefont {S.}~\bibnamefont {Lin}}, \ and\ \bibinfo {author}
  {\bibfnamefont {J.}~\bibnamefont {Wang}},\ }\href@noop {} {\bibfield
  {journal} {\bibinfo  {journal} {Quantum Information Processing}\ }\textbf
  {\bibinfo {volume} {18}},\ \bibinfo {pages} {249} (\bibinfo {year}
  {2019}{\natexlab{a}})}\BibitemShut {NoStop}%
\bibitem [{\citenamefont {Yu}\ \emph {et~al.}(2019{\natexlab{b}})\citenamefont
  {Yu}, \citenamefont {Gao},\ and\ \citenamefont {Wen}}]{yu2019improved}%
  \BibitemOpen
  \bibfield  {author} {\bibinfo {author} {\bibfnamefont {C.-H.}\ \bibnamefont
  {Yu}}, \bibinfo {author} {\bibfnamefont {F.}~\bibnamefont {Gao}}, \ and\
  \bibinfo {author} {\bibfnamefont {Q.}~\bibnamefont {Wen}},\ }\href@noop {}
  {\bibfield  {journal} {\bibinfo  {journal} {IEEE Transactions on Knowledge
  and Data Engineering}\ } (\bibinfo {year} {2019}{\natexlab{b}})}\BibitemShut
  {NoStop}%
\bibitem [{\citenamefont {Pan}\ \emph {et~al.}(2020)\citenamefont {Pan},
  \citenamefont {Wan}, \citenamefont {Liu}, \citenamefont {Wang}, \citenamefont
  {Qin}, \citenamefont {Wen},\ and\ \citenamefont {Gao}}]{pan2020improved}%
  \BibitemOpen
  \bibfield  {author} {\bibinfo {author} {\bibfnamefont {S.-J.}\ \bibnamefont
  {Pan}}, \bibinfo {author} {\bibfnamefont {L.-C.}\ \bibnamefont {Wan}},
  \bibinfo {author} {\bibfnamefont {H.-L.}\ \bibnamefont {Liu}}, \bibinfo
  {author} {\bibfnamefont {Q.-L.}\ \bibnamefont {Wang}}, \bibinfo {author}
  {\bibfnamefont {S.-J.}\ \bibnamefont {Qin}}, \bibinfo {author} {\bibfnamefont
  {Q.-Y.}\ \bibnamefont {Wen}}, \ and\ \bibinfo {author} {\bibfnamefont
  {F.}~\bibnamefont {Gao}},\ }\href@noop {} {\bibfield  {journal} {\bibinfo
  {journal} {Physical Review A}\ }\textbf {\bibinfo {volume} {102}},\ \bibinfo
  {pages} {052402} (\bibinfo {year} {2020})}\BibitemShut {NoStop}%
\bibitem [{\citenamefont {Montanaro}(2016)}]{alg}%
  \BibitemOpen
  \bibfield  {author} {\bibinfo {author} {\bibfnamefont {A.}~\bibnamefont
  {Montanaro}},\ }\href@noop {} {\bibfield  {journal} {\bibinfo  {journal} {npj
  Quantum Information}\ }\textbf {\bibinfo {volume} {2}},\ \bibinfo {pages}
  {15023} (\bibinfo {year} {2016})}\BibitemShut {NoStop}%
\bibitem [{\citenamefont {Kailath}\ and\ \citenamefont
  {Sayed}(1999)}]{kailath1999fast}%
  \BibitemOpen
  \bibfield  {author} {\bibinfo {author} {\bibfnamefont {T.}~\bibnamefont
  {Kailath}}\ and\ \bibinfo {author} {\bibfnamefont {A.~H.}\ \bibnamefont
  {Sayed}},\ }\href@noop {} {\emph {\bibinfo {title} {Fast reliable algorithms
  for matrices with structure}}}\ (\bibinfo  {publisher} {SIAM},\ \bibinfo
  {year} {1999})\BibitemShut {NoStop}%
\bibitem [{\citenamefont {Pan}(2001)}]{pan2001structured}%
  \BibitemOpen
  \bibfield  {author} {\bibinfo {author} {\bibfnamefont {V.}~\bibnamefont
  {Pan}},\ }\href@noop {} {\emph {\bibinfo {title} {Structured Matrices and
  Polynomials: Unified Superfast Algorithms}}}\ (\bibinfo  {publisher}
  {Springer Science \& Business Media},\ \bibinfo {year} {2001})\BibitemShut
  {NoStop}%
\bibitem [{\citenamefont {Peller}(2012)}]{peller2012hankel}%
  \BibitemOpen
  \bibfield  {author} {\bibinfo {author} {\bibfnamefont {V.}~\bibnamefont
  {Peller}},\ }\href@noop {} {\emph {\bibinfo {title} {Hankel operators and
  their applications}}}\ (\bibinfo  {publisher} {Springer Science \& Business
  Media},\ \bibinfo {year} {2012})\BibitemShut {NoStop}%
\bibitem [{\citenamefont {Cao}\ \emph {et~al.}(2013)\citenamefont {Cao},
  \citenamefont {Papageorgiou}, \citenamefont {Petras}, \citenamefont {Traub},\
  and\ \citenamefont {Kais}}]{pos}%
  \BibitemOpen
  \bibfield  {author} {\bibinfo {author} {\bibfnamefont {Y.}~\bibnamefont
  {Cao}}, \bibinfo {author} {\bibfnamefont {A.}~\bibnamefont {Papageorgiou}},
  \bibinfo {author} {\bibfnamefont {I.}~\bibnamefont {Petras}}, \bibinfo
  {author} {\bibfnamefont {J.}~\bibnamefont {Traub}}, \ and\ \bibinfo {author}
  {\bibfnamefont {S.}~\bibnamefont {Kais}},\ }\href@noop {} {\bibfield
  {journal} {\bibinfo  {journal} {New Journal of Physics}\ }\textbf {\bibinfo
  {volume} {15}},\ \bibinfo {pages} {013021} (\bibinfo {year}
  {2013})}\BibitemShut {NoStop}%
\bibitem [{\citenamefont {Mahasinghe}\ and\ \citenamefont {Wang}(2016)}]{Tb}%
  \BibitemOpen
  \bibfield  {author} {\bibinfo {author} {\bibfnamefont {A.}~\bibnamefont
  {Mahasinghe}}\ and\ \bibinfo {author} {\bibfnamefont {J.~B.}\ \bibnamefont
  {Wang}},\ }\href@noop {} {\bibfield  {journal} {\bibinfo  {journal} {Journal
  of Physics A: Mathematical and Theoretical}\ }\textbf {\bibinfo {volume}
  {49}},\ \bibinfo {pages} {275301} (\bibinfo {year} {2016})}\BibitemShut
  {NoStop}%
\bibitem [{\citenamefont {Zhou}\ and\ \citenamefont
  {Wang}(2017)}]{zhou2017efficient}%
  \BibitemOpen
  \bibfield  {author} {\bibinfo {author} {\bibfnamefont {S.~S.}\ \bibnamefont
  {Zhou}}\ and\ \bibinfo {author} {\bibfnamefont {J.~B.}\ \bibnamefont
  {Wang}},\ }\href@noop {} {\bibfield  {journal} {\bibinfo  {journal} {Royal
  Society Open Science}\ }\textbf {\bibinfo {volume} {4}},\ \bibinfo {pages}
  {160906} (\bibinfo {year} {2017})}\BibitemShut {NoStop}%
\bibitem [{\citenamefont {Wan}\ \emph {et~al.}(2018)\citenamefont {Wan},
  \citenamefont {Yu}, \citenamefont {Pan}, \citenamefont {Gao}, \citenamefont
  {Wen},\ and\ \citenamefont {Qin}}]{wan2018asymptotic}%
  \BibitemOpen
  \bibfield  {author} {\bibinfo {author} {\bibfnamefont {L.-C.}\ \bibnamefont
  {Wan}}, \bibinfo {author} {\bibfnamefont {C.-H.}\ \bibnamefont {Yu}},
  \bibinfo {author} {\bibfnamefont {S.-J.}\ \bibnamefont {Pan}}, \bibinfo
  {author} {\bibfnamefont {F.}~\bibnamefont {Gao}}, \bibinfo {author}
  {\bibfnamefont {Q.-Y.}\ \bibnamefont {Wen}}, \ and\ \bibinfo {author}
  {\bibfnamefont {S.-J.}\ \bibnamefont {Qin}},\ }\href@noop {} {\bibfield
  {journal} {\bibinfo  {journal} {Physical Review A}\ }\textbf {\bibinfo
  {volume} {97}},\ \bibinfo {pages} {062322} (\bibinfo {year}
  {2018})}\BibitemShut {NoStop}%
\bibitem [{\citenamefont {Low}\ and\ \citenamefont
  {Chuang}(2019)}]{low2019hamiltonian}%
  \BibitemOpen
  \bibfield  {author} {\bibinfo {author} {\bibfnamefont {G.~H.}\ \bibnamefont
  {Low}}\ and\ \bibinfo {author} {\bibfnamefont {I.~L.}\ \bibnamefont
  {Chuang}},\ }\href@noop {} {\bibfield  {journal} {\bibinfo  {journal}
  {Quantum}\ }\textbf {\bibinfo {volume} {3}},\ \bibinfo {pages} {163}
  (\bibinfo {year} {2019})}\BibitemShut {NoStop}%
\bibitem [{\citenamefont {Chakraborty}\ \emph {et~al.}(2019)\citenamefont
  {Chakraborty}, \citenamefont {Gily{\'e}n},\ and\ \citenamefont
  {Jeffery}}]{chakraborty2019power}%
  \BibitemOpen
  \bibfield  {author} {\bibinfo {author} {\bibfnamefont {S.}~\bibnamefont
  {Chakraborty}}, \bibinfo {author} {\bibfnamefont {A.}~\bibnamefont
  {Gily{\'e}n}}, \ and\ \bibinfo {author} {\bibfnamefont {S.}~\bibnamefont
  {Jeffery}},\ }in\ \href@noop {} {\emph {\bibinfo {booktitle} {46th
  International Colloquium on Automata, Languages, and Programming (ICALP
  2019)}}}\ (\bibinfo {organization} {Schloss Dagstuhl-Leibniz-Zentrum fuer
  Informatik},\ \bibinfo {year} {2019})\BibitemShut {NoStop}%
\bibitem [{\citenamefont {van Apeldoorn}\ and\ \citenamefont
  {Gily{\'e}n}(2019)}]{van2019improvements}%
  \BibitemOpen
  \bibfield  {author} {\bibinfo {author} {\bibfnamefont {J.}~\bibnamefont {van
  Apeldoorn}}\ and\ \bibinfo {author} {\bibfnamefont {A.}~\bibnamefont
  {Gily{\'e}n}},\ }in\ \href@noop {} {\emph {\bibinfo {booktitle} {46th
  International Colloquium on Automata, Languages, and Programming (ICALP
  2019)}}}\ (\bibinfo {organization} {Schloss Dagstuhl-Leibniz-Zentrum fuer
  Informatik},\ \bibinfo {year} {2019})\BibitemShut {NoStop}%
\bibitem [{\citenamefont {Kerenidis}\ and\ \citenamefont
  {Prakash}(2020{\natexlab{a}})}]{kerenidis2020quantum}%
  \BibitemOpen
  \bibfield  {author} {\bibinfo {author} {\bibfnamefont {I.}~\bibnamefont
  {Kerenidis}}\ and\ \bibinfo {author} {\bibfnamefont {A.}~\bibnamefont
  {Prakash}},\ }\href@noop {} {\bibfield  {journal} {\bibinfo  {journal}
  {Physical Review A}\ }\textbf {\bibinfo {volume} {101}},\ \bibinfo {pages}
  {022316} (\bibinfo {year} {2020}{\natexlab{a}})}\BibitemShut {NoStop}%
\bibitem [{\citenamefont {Childs}\ \emph {et~al.}(2017)\citenamefont {Childs},
  \citenamefont {Kothari},\ and\ \citenamefont {Somma}}]{Childs2017Quantum}%
  \BibitemOpen
  \bibfield  {author} {\bibinfo {author} {\bibfnamefont {A.~M.}\ \bibnamefont
  {Childs}}, \bibinfo {author} {\bibfnamefont {R.}~\bibnamefont {Kothari}}, \
  and\ \bibinfo {author} {\bibfnamefont {R.~D.}\ \bibnamefont {Somma}},\
  }\href@noop {} {\bibfield  {journal} {\bibinfo  {journal} {Siam Journal on
  Computing}\ }\textbf {\bibinfo {volume} {46}} (\bibinfo {year}
  {2017})}\BibitemShut {NoStop}%
\bibitem [{\citenamefont {Gray}\ \emph {et~al.}(2006)\citenamefont {Gray} \emph
  {et~al.}}]{Toe}%
  \BibitemOpen
  \bibfield  {author} {\bibinfo {author} {\bibfnamefont {R.~M.}\ \bibnamefont
  {Gray}} \emph {et~al.},\ }\href@noop {} {\bibfield  {journal} {\bibinfo
  {journal} {Foundations and Trends{\textregistered} in Communications and
  Information Theory}\ }\textbf {\bibinfo {volume} {2}},\ \bibinfo {pages}
  {155} (\bibinfo {year} {2006})}\BibitemShut {NoStop}%
\bibitem [{\citenamefont {Henriques}\ \emph {et~al.}(2014)\citenamefont
  {Henriques}, \citenamefont {Caseiro}, \citenamefont {Martins},\ and\
  \citenamefont {Batista}}]{henriques2014high}%
  \BibitemOpen
  \bibfield  {author} {\bibinfo {author} {\bibfnamefont {J.~F.}\ \bibnamefont
  {Henriques}}, \bibinfo {author} {\bibfnamefont {R.}~\bibnamefont {Caseiro}},
  \bibinfo {author} {\bibfnamefont {P.}~\bibnamefont {Martins}}, \ and\
  \bibinfo {author} {\bibfnamefont {J.}~\bibnamefont {Batista}},\ }\href@noop
  {} {\bibfield  {journal} {\bibinfo  {journal} {IEEE transactions on pattern
  analysis and machine intelligence}\ }\textbf {\bibinfo {volume} {37}},\
  \bibinfo {pages} {583} (\bibinfo {year} {2014})}\BibitemShut {NoStop}%
\bibitem [{\citenamefont {Smith}\ \emph {et~al.}(1985)\citenamefont {Smith},
  \citenamefont {Smith},\ and\ \citenamefont {Smith}}]{smith1985numerical}%
  \BibitemOpen
  \bibfield  {author} {\bibinfo {author} {\bibfnamefont {G.~D.}\ \bibnamefont
  {Smith}}, \bibinfo {author} {\bibfnamefont {G.~D.}\ \bibnamefont {Smith}}, \
  and\ \bibinfo {author} {\bibfnamefont {G.~D.~S.}\ \bibnamefont {Smith}},\
  }\href@noop {} {\emph {\bibinfo {title} {Numerical solution of partial
  differential equations: finite difference methods}}}\ (\bibinfo  {publisher}
  {Oxford university press},\ \bibinfo {year} {1985})\BibitemShut {NoStop}%
\bibitem [{\citenamefont {Chan}\ and\ \citenamefont {Ng}(1996)}]{CGT}%
  \BibitemOpen
  \bibfield  {author} {\bibinfo {author} {\bibfnamefont {R.~H.}\ \bibnamefont
  {Chan}}\ and\ \bibinfo {author} {\bibfnamefont {M.~K.}\ \bibnamefont {Ng}},\
  }\href@noop {} {\bibfield  {journal} {\bibinfo  {journal} {SIAM review}\
  }\textbf {\bibinfo {volume} {38}},\ \bibinfo {pages} {427} (\bibinfo {year}
  {1996})}\BibitemShut {NoStop}%
\bibitem [{\citenamefont {Ye}\ \emph {et~al.}(2018)\citenamefont {Ye},
  \citenamefont {Han},\ and\ \citenamefont {Cha}}]{ye2018deep}%
  \BibitemOpen
  \bibfield  {author} {\bibinfo {author} {\bibfnamefont {J.~C.}\ \bibnamefont
  {Ye}}, \bibinfo {author} {\bibfnamefont {Y.}~\bibnamefont {Han}}, \ and\
  \bibinfo {author} {\bibfnamefont {E.}~\bibnamefont {Cha}},\ }\href@noop {}
  {\bibfield  {journal} {\bibinfo  {journal} {SIAM Journal on Imaging
  Sciences}\ }\textbf {\bibinfo {volume} {11}},\ \bibinfo {pages} {991}
  (\bibinfo {year} {2018})}\BibitemShut {NoStop}%
\bibitem [{\citenamefont {Haupt}\ \emph {et~al.}(2010)\citenamefont {Haupt},
  \citenamefont {Bajwa}, \citenamefont {Raz},\ and\ \citenamefont
  {Nowak}}]{haupt2010toeplitz}%
  \BibitemOpen
  \bibfield  {author} {\bibinfo {author} {\bibfnamefont {J.}~\bibnamefont
  {Haupt}}, \bibinfo {author} {\bibfnamefont {W.~U.}\ \bibnamefont {Bajwa}},
  \bibinfo {author} {\bibfnamefont {G.}~\bibnamefont {Raz}}, \ and\ \bibinfo
  {author} {\bibfnamefont {R.}~\bibnamefont {Nowak}},\ }\href@noop {}
  {\bibfield  {journal} {\bibinfo  {journal} {IEEE transactions on information
  theory}\ }\textbf {\bibinfo {volume} {56}},\ \bibinfo {pages} {5862}
  (\bibinfo {year} {2010})}\BibitemShut {NoStop}%
\bibitem [{\citenamefont {Van~Apeldoorn}\ \emph {et~al.}(2017)\citenamefont
  {Van~Apeldoorn}, \citenamefont {Gily{\'e}n}, \citenamefont {Gribling},\ and\
  \citenamefont {de~Wolf}}]{van2017quantum}%
  \BibitemOpen
  \bibfield  {author} {\bibinfo {author} {\bibfnamefont {J.}~\bibnamefont
  {Van~Apeldoorn}}, \bibinfo {author} {\bibfnamefont {A.}~\bibnamefont
  {Gily{\'e}n}}, \bibinfo {author} {\bibfnamefont {S.}~\bibnamefont
  {Gribling}}, \ and\ \bibinfo {author} {\bibfnamefont {R.}~\bibnamefont
  {de~Wolf}},\ }in\ \href@noop {} {\emph {\bibinfo {booktitle} {Foundations of
  Computer Science (FOCS), 2017 IEEE 58th Annual Symposium on}}}\ (\bibinfo
  {organization} {IEEE},\ \bibinfo {year} {2017})\ pp.\ \bibinfo {pages}
  {403--414}\BibitemShut {NoStop}%
\bibitem [{\citenamefont {Kerenidis}\ \emph {et~al.}(2019)\citenamefont
  {Kerenidis}, \citenamefont {Landman}, \citenamefont {Luongo},\ and\
  \citenamefont {Prakash}}]{kerenidis2019q}%
  \BibitemOpen
  \bibfield  {author} {\bibinfo {author} {\bibfnamefont {I.}~\bibnamefont
  {Kerenidis}}, \bibinfo {author} {\bibfnamefont {J.}~\bibnamefont {Landman}},
  \bibinfo {author} {\bibfnamefont {A.}~\bibnamefont {Luongo}}, \ and\ \bibinfo
  {author} {\bibfnamefont {A.}~\bibnamefont {Prakash}},\ }in\ \href@noop {}
  {\emph {\bibinfo {booktitle} {Advances in Neural Information Processing
  Systems}}}\ (\bibinfo {year} {2019})\ pp.\ \bibinfo {pages}
  {4136--4146}\BibitemShut {NoStop}%
\bibitem [{\citenamefont {Shao}(2019)}]{shao2019quantum}%
  \BibitemOpen
  \bibfield  {author} {\bibinfo {author} {\bibfnamefont {C.}~\bibnamefont
  {Shao}},\ }\href@noop {} {\bibfield  {journal} {\bibinfo  {journal} {Journal
  of Physics A: Mathematical and Theoretical}\ } (\bibinfo {year}
  {2019})}\BibitemShut {NoStop}%
\bibitem [{\citenamefont {Kerenidis}\ and\ \citenamefont
  {Prakash}(2020{\natexlab{b}})}]{kerenidis2020quantumACM}%
  \BibitemOpen
  \bibfield  {author} {\bibinfo {author} {\bibfnamefont {I.}~\bibnamefont
  {Kerenidis}}\ and\ \bibinfo {author} {\bibfnamefont {A.}~\bibnamefont
  {Prakash}},\ }\href@noop {} {\bibfield  {journal} {\bibinfo  {journal} {ACM
  Transactions on Quantum Computing}\ }\textbf {\bibinfo {volume} {1}},\
  \bibinfo {pages} {1} (\bibinfo {year} {2020}{\natexlab{b}})}\BibitemShut
  {NoStop}%
\bibitem [{\citenamefont {Grover}\ and\ \citenamefont
  {Rudolph}(2002)}]{prepare}%
  \BibitemOpen
  \bibfield  {author} {\bibinfo {author} {\bibfnamefont {L.}~\bibnamefont
  {Grover}}\ and\ \bibinfo {author} {\bibfnamefont {T.}~\bibnamefont
  {Rudolph}},\ }\href@noop {} {\bibfield  {journal} {\bibinfo  {journal} {arXiv
  preprint quant-ph/0208112}\ } (\bibinfo {year} {2002})}\BibitemShut {NoStop}%
\bibitem [{\citenamefont {Soklakov}\ and\ \citenamefont
  {Schack}(2006)}]{prepare1}%
  \BibitemOpen
  \bibfield  {author} {\bibinfo {author} {\bibfnamefont {A.~N.}\ \bibnamefont
  {Soklakov}}\ and\ \bibinfo {author} {\bibfnamefont {R.}~\bibnamefont
  {Schack}},\ }\href@noop {} {\bibfield  {journal} {\bibinfo  {journal}
  {Physical Review A}\ }\textbf {\bibinfo {volume} {73}},\ \bibinfo {pages}
  {012307} (\bibinfo {year} {2006})}\BibitemShut {NoStop}%
\bibitem [{\citenamefont {Berry}\ and\ \citenamefont
  {Childs}(2012)}]{berry2012black}%
  \BibitemOpen
  \bibfield  {author} {\bibinfo {author} {\bibfnamefont {D.~W.}\ \bibnamefont
  {Berry}}\ and\ \bibinfo {author} {\bibfnamefont {A.~M.}\ \bibnamefont
  {Childs}},\ }\href@noop {} {\bibfield  {journal} {\bibinfo  {journal}
  {Quantum Information \& Computation}\ }\textbf {\bibinfo {volume} {12}},\
  \bibinfo {pages} {29} (\bibinfo {year} {2012})}\BibitemShut {NoStop}%
\bibitem [{\citenamefont {Sanders}\ \emph {et~al.}(2019)\citenamefont
  {Sanders}, \citenamefont {Low}, \citenamefont {Scherer},\ and\ \citenamefont
  {Berry}}]{sanders2019black}%
  \BibitemOpen
  \bibfield  {author} {\bibinfo {author} {\bibfnamefont {Y.~R.}\ \bibnamefont
  {Sanders}}, \bibinfo {author} {\bibfnamefont {G.~H.}\ \bibnamefont {Low}},
  \bibinfo {author} {\bibfnamefont {A.}~\bibnamefont {Scherer}}, \ and\
  \bibinfo {author} {\bibfnamefont {D.~W.}\ \bibnamefont {Berry}},\ }\href@noop
  {} {\bibfield  {journal} {\bibinfo  {journal} {Physical review letters}\
  }\textbf {\bibinfo {volume} {122}},\ \bibinfo {pages} {020502} (\bibinfo
  {year} {2019})}\BibitemShut {NoStop}%
\bibitem [{\citenamefont {Kerenidis}\ and\ \citenamefont
  {Prakash}()}]{kerenidis2016quantum}%
  \BibitemOpen
  \bibfield  {author} {\bibinfo {author} {\bibfnamefont {I.}~\bibnamefont
  {Kerenidis}}\ and\ \bibinfo {author} {\bibfnamefont {A.}~\bibnamefont
  {Prakash}},\ }in\ \href {\doibase 10.4230/LIPIcs.ITCS.2017.49} {\emph
  {\bibinfo {booktitle} {8th Innovations in Theoretical Computer Science
  Conference (ITCS 2017)}}},\ Vol.~\bibinfo {volume} {67},\ pp.\ \bibinfo
  {pages} {49:1--49:21}\BibitemShut {NoStop}%
\bibitem [{\citenamefont {Grover}(2000)}]{grover2000synthesis}%
  \BibitemOpen
  \bibfield  {author} {\bibinfo {author} {\bibfnamefont {L.~K.}\ \bibnamefont
  {Grover}},\ }\href@noop {} {\bibfield  {journal} {\bibinfo  {journal}
  {Physical review letters}\ }\textbf {\bibinfo {volume} {85}},\ \bibinfo
  {pages} {1334} (\bibinfo {year} {2000})}\BibitemShut {NoStop}%
\bibitem [{\citenamefont {Yoder}\ \emph {et~al.}(2014)\citenamefont {Yoder},
  \citenamefont {Low},\ and\ \citenamefont {Chuang}}]{yoder2014fixed}%
  \BibitemOpen
  \bibfield  {author} {\bibinfo {author} {\bibfnamefont {T.~J.}\ \bibnamefont
  {Yoder}}, \bibinfo {author} {\bibfnamefont {G.~H.}\ \bibnamefont {Low}}, \
  and\ \bibinfo {author} {\bibfnamefont {I.~L.}\ \bibnamefont {Chuang}},\
  }\href@noop {} {\bibfield  {journal} {\bibinfo  {journal} {Physical review
  letters}\ }\textbf {\bibinfo {volume} {113}},\ \bibinfo {pages} {210501}
  (\bibinfo {year} {2014})}\BibitemShut {NoStop}%
\bibitem [{\citenamefont {Berry}\ \emph {et~al.}(2015)\citenamefont {Berry},
  \citenamefont {Childs}, \citenamefont {Cleve}, \citenamefont {Kothari},\ and\
  \citenamefont {Somma}}]{berry2015simulating}%
  \BibitemOpen
  \bibfield  {author} {\bibinfo {author} {\bibfnamefont {D.~W.}\ \bibnamefont
  {Berry}}, \bibinfo {author} {\bibfnamefont {A.~M.}\ \bibnamefont {Childs}},
  \bibinfo {author} {\bibfnamefont {R.}~\bibnamefont {Cleve}}, \bibinfo
  {author} {\bibfnamefont {R.}~\bibnamefont {Kothari}}, \ and\ \bibinfo
  {author} {\bibfnamefont {R.~D.}\ \bibnamefont {Somma}},\ }\href@noop {}
  {\bibfield  {journal} {\bibinfo  {journal} {Physical review letters}\
  }\textbf {\bibinfo {volume} {114}},\ \bibinfo {pages} {090502} (\bibinfo
  {year} {2015})}\BibitemShut {NoStop}%
\bibitem [{\citenamefont {Gui-Lu}\ \emph {et~al.}(2009)\citenamefont {Gui-Lu},
  \citenamefont {Yang},\ and\ \citenamefont {Chuan}}]{gui2009allowable}%
  \BibitemOpen
  \bibfield  {author} {\bibinfo {author} {\bibfnamefont {L.}~\bibnamefont
  {Gui-Lu}}, \bibinfo {author} {\bibfnamefont {L.}~\bibnamefont {Yang}}, \ and\
  \bibinfo {author} {\bibfnamefont {W.}~\bibnamefont {Chuan}},\ }\href@noop {}
  {\bibfield  {journal} {\bibinfo  {journal} {Communications in Theoretical
  Physics}\ }\textbf {\bibinfo {volume} {51}},\ \bibinfo {pages} {65} (\bibinfo
  {year} {2009})}\BibitemShut {NoStop}%
\bibitem [{\citenamefont {Yu}\ \emph {et~al.}(2019{\natexlab{c}})\citenamefont
  {Yu}, \citenamefont {Gao}, \citenamefont {Liu}, \citenamefont {Huynh},
  \citenamefont {Reynolds},\ and\ \citenamefont {Wang}}]{yu2019quantum}%
  \BibitemOpen
  \bibfield  {author} {\bibinfo {author} {\bibfnamefont {C.-H.}\ \bibnamefont
  {Yu}}, \bibinfo {author} {\bibfnamefont {F.}~\bibnamefont {Gao}}, \bibinfo
  {author} {\bibfnamefont {C.}~\bibnamefont {Liu}}, \bibinfo {author}
  {\bibfnamefont {D.}~\bibnamefont {Huynh}}, \bibinfo {author} {\bibfnamefont
  {M.}~\bibnamefont {Reynolds}}, \ and\ \bibinfo {author} {\bibfnamefont
  {J.}~\bibnamefont {Wang}},\ }\href@noop {} {\bibfield  {journal} {\bibinfo
  {journal} {Physical Review A}\ }\textbf {\bibinfo {volume} {99}},\ \bibinfo
  {pages} {022301} (\bibinfo {year} {2019}{\natexlab{c}})}\BibitemShut
  {NoStop}%
\bibitem [{\citenamefont {Ambainis}(2012)}]{ambainis2012variable}%
  \BibitemOpen
  \bibfield  {author} {\bibinfo {author} {\bibfnamefont {A.}~\bibnamefont
  {Ambainis}},\ }in\ \href@noop {} {\emph {\bibinfo {booktitle} {STACS'12 (29th
  Symposium on Theoretical Aspects of Computer Science)}}},\ Vol.~\bibinfo
  {volume} {14}\ (\bibinfo {organization} {LIPIcs},\ \bibinfo {year} {2012})\
  pp.\ \bibinfo {pages} {636--647}\BibitemShut {NoStop}%
\bibitem [{\citenamefont {Habermann}(2018)}]{habermann2018explicit}%
  \BibitemOpen
  \bibfield  {author} {\bibinfo {author} {\bibfnamefont {K.}~\bibnamefont
  {Habermann}},\ }\href@noop {} {\bibfield  {journal} {\bibinfo  {journal}
  {arXiv preprint arXiv:1808.02880}\ } (\bibinfo {year} {2018})}\BibitemShut
  {NoStop}%
\bibitem [{\citenamefont {Widom}(1966)}]{widom1966hankel}%
  \BibitemOpen
  \bibfield  {author} {\bibinfo {author} {\bibfnamefont {H.}~\bibnamefont
  {Widom}},\ }\href@noop {} {\bibfield  {journal} {\bibinfo  {journal}
  {Transactions of the American Mathematical Society}\ }\textbf {\bibinfo
  {volume} {121}},\ \bibinfo {pages} {1} (\bibinfo {year} {1966})}\BibitemShut
  {NoStop}%
\bibitem [{\citenamefont {Sch{\"o}nhage}(1985)}]{schonhage1985quasi}%
  \BibitemOpen
  \bibfield  {author} {\bibinfo {author} {\bibfnamefont {A.}~\bibnamefont
  {Sch{\"o}nhage}},\ }\href@noop {} {\bibfield  {journal} {\bibinfo  {journal}
  {Journal of Complexity}\ }\textbf {\bibinfo {volume} {1}},\ \bibinfo {pages}
  {118} (\bibinfo {year} {1985})}\BibitemShut {NoStop}%
\bibitem [{\citenamefont {Haykin}(2005)}]{haykin2005adaptive}%
  \BibitemOpen
  \bibfield  {author} {\bibinfo {author} {\bibfnamefont {S.~S.}\ \bibnamefont
  {Haykin}},\ }\href@noop {} {\emph {\bibinfo {title} {Adaptive filter
  theory}}}\ (\bibinfo  {publisher} {Pearson Education India},\ \bibinfo {year}
  {2005})\BibitemShut {NoStop}%
\bibitem [{\citenamefont {Aharonov}\ \emph {et~al.}(2009)\citenamefont
  {Aharonov}, \citenamefont {Jones},\ and\ \citenamefont
  {Landau}}]{aharonov2009polynomial}%
  \BibitemOpen
  \bibfield  {author} {\bibinfo {author} {\bibfnamefont {D.}~\bibnamefont
  {Aharonov}}, \bibinfo {author} {\bibfnamefont {V.}~\bibnamefont {Jones}}, \
  and\ \bibinfo {author} {\bibfnamefont {Z.}~\bibnamefont {Landau}},\
  }\href@noop {} {\bibfield  {journal} {\bibinfo  {journal} {Algorithmica}\
  }\textbf {\bibinfo {volume} {55}},\ \bibinfo {pages} {395} (\bibinfo {year}
  {2009})}\BibitemShut {NoStop}%
\bibitem [{\citenamefont {Wu}\ \emph {et~al.}(2021)\citenamefont {Wu},
  \citenamefont {Ray}, \citenamefont {Zhao}, \citenamefont {Sun},\ and\
  \citenamefont {Rebentrost}}]{wu2021quantum}%
  \BibitemOpen
  \bibfield  {author} {\bibinfo {author} {\bibfnamefont {B.}~\bibnamefont
  {Wu}}, \bibinfo {author} {\bibfnamefont {M.}~\bibnamefont {Ray}}, \bibinfo
  {author} {\bibfnamefont {L.}~\bibnamefont {Zhao}}, \bibinfo {author}
  {\bibfnamefont {X.}~\bibnamefont {Sun}}, \ and\ \bibinfo {author}
  {\bibfnamefont {P.}~\bibnamefont {Rebentrost}},\ }\href@noop {} {\bibfield
  {journal} {\bibinfo  {journal} {Physical Review A}\ }\textbf {\bibinfo
  {volume} {103}},\ \bibinfo {pages} {042422} (\bibinfo {year}
  {2021})}\BibitemShut {NoStop}%
\bibitem [{\citenamefont {Shende}\ \emph {et~al.}(2006)\citenamefont {Shende},
  \citenamefont {Bullock},\ and\ \citenamefont {Markov}}]{shende2006synthesis}%
  \BibitemOpen
  \bibfield  {author} {\bibinfo {author} {\bibfnamefont {V.~V.}\ \bibnamefont
  {Shende}}, \bibinfo {author} {\bibfnamefont {S.~S.}\ \bibnamefont {Bullock}},
  \ and\ \bibinfo {author} {\bibfnamefont {I.~L.}\ \bibnamefont {Markov}},\
  }\href@noop {} {\bibfield  {journal} {\bibinfo  {journal} {IEEE Transactions
  on Computer-Aided Design of Integrated Circuits and Systems}\ }\textbf
  {\bibinfo {volume} {25}},\ \bibinfo {pages} {1000} (\bibinfo {year}
  {2006})}\BibitemShut {NoStop}%
\bibitem [{\citenamefont {Tang}(2019)}]{tang2019quantum}%
  \BibitemOpen
  \bibfield  {author} {\bibinfo {author} {\bibfnamefont {E.}~\bibnamefont
  {Tang}},\ }in\ \href@noop {} {\emph {\bibinfo {booktitle} {Proceedings of the
  51st Annual ACM SIGACT Symposium on Theory of Computing}}}\ (\bibinfo {year}
  {2019})\ pp.\ \bibinfo {pages} {217--228}\BibitemShut {NoStop}%
\bibitem [{\citenamefont {Cuccaro}\ \emph {et~al.}(2004)\citenamefont
  {Cuccaro}, \citenamefont {Draper}, \citenamefont {Kutin},\ and\ \citenamefont
  {Moulton}}]{cuccaro2004new}%
  \BibitemOpen
  \bibfield  {author} {\bibinfo {author} {\bibfnamefont {S.~A.}\ \bibnamefont
  {Cuccaro}}, \bibinfo {author} {\bibfnamefont {T.~G.}\ \bibnamefont {Draper}},
  \bibinfo {author} {\bibfnamefont {S.~A.}\ \bibnamefont {Kutin}}, \ and\
  \bibinfo {author} {\bibfnamefont {D.~P.}\ \bibnamefont {Moulton}},\
  }\href@noop {} {\bibfield  {journal} {\bibinfo  {journal} {arXiv preprint
  quant-ph/0410184}\ } (\bibinfo {year} {2004})}\BibitemShut {NoStop}%
\bibitem [{\citenamefont {Li}\ \emph {et~al.}(2020)\citenamefont {Li},
  \citenamefont {Fan}, \citenamefont {Xia}, \citenamefont {Peng},\ and\
  \citenamefont {Long}}]{li2020efficient}%
  \BibitemOpen
  \bibfield  {author} {\bibinfo {author} {\bibfnamefont {H.-S.}\ \bibnamefont
  {Li}}, \bibinfo {author} {\bibfnamefont {P.}~\bibnamefont {Fan}}, \bibinfo
  {author} {\bibfnamefont {H.}~\bibnamefont {Xia}}, \bibinfo {author}
  {\bibfnamefont {H.}~\bibnamefont {Peng}}, \ and\ \bibinfo {author}
  {\bibfnamefont {G.-L.}\ \bibnamefont {Long}},\ }\href@noop {} {\bibfield
  {journal} {\bibinfo  {journal} {SCIENCE CHINA Physics, Mechanics \&
  Astronomy}\ }\textbf {\bibinfo {volume} {63}},\ \bibinfo {pages} {1}
  (\bibinfo {year} {2020})}\BibitemShut {NoStop}%
\bibitem [{\citenamefont {Long}(2001)}]{long2001grover}%
  \BibitemOpen
  \bibfield  {author} {\bibinfo {author} {\bibfnamefont {G.-L.}\ \bibnamefont
  {Long}},\ }\href@noop {} {\bibfield  {journal} {\bibinfo  {journal} {Physical
  Review A}\ }\textbf {\bibinfo {volume} {64}},\ \bibinfo {pages} {022307}
  (\bibinfo {year} {2001})}\BibitemShut {NoStop}%
\end{thebibliography}%

\end{document}